\documentclass[11pt]{report}
\usepackage{setspace}
\usepackage{fancyhdr}
\usepackage{graphicx}
\usepackage{appendix}
\usepackage{amsmath, amsthm, amssymb}
\usepackage{amsfonts}
\usepackage{mathrsfs}
\usepackage{placeins}
\usepackage{multicol}
\usepackage{multirow}
\usepackage{subfig}
\usepackage{algorithm}
\usepackage{algorithmic}
\usepackage{anysize}
\usepackage{color}

\newtheorem{lem}{Lemma}
\marginsize{2cm}{2cm}{2cm}{2cm} 

\begin{document}

\begin{titlepage}
\begin{center}
\null \vspace{\stretch{0.2}}
\renewcommand{\baselinestretch}{1}
\textsc{\large{\large{{Overcoming the ill-posedness through discretization in vector tomography:}}}\\
Reconstruction of irrotational vector fields}\\

\vspace{\stretch{1}} Technical (PhD Transfer) Report

 \vspace{\stretch{1}} January 2011

\vspace{\stretch{1}} \textsc{Student: Alexandra Koulouri}\\
\textsc{Supervisor: Prof. M. Petrou }\\
\null \vspace{\stretch{0.2}}
\renewcommand{\baselinestretch}{1}
\vspace{\stretch{0.1}}
Communications \& Signal Processing Group\\
Department of Electrical \& Electronics Engineering\\
Imperial College London\\
\end{center}
\end{titlepage}

\pagenumbering{roman}

\tableofcontents

\pagenumbering{arabic}
\chapter{Introduction} \label{ch:intro}
Vector field tomographic methods intend to reconstruct and visualize
a vector field in a bounded domain by measuring line integrals of
projections of this vector field.

In particular, we have to deal with an inverse problem of recovering
a vector function from boundary measurements. As the majority of
inverse problems, vector field method is ill posed in the continuous
domain and therefore further assumptions, measurements and
constraints should be imposed for the full vector field recovery.
The reconstruction idea in the discrete domain relies on solving a
numerical system of linear equations which derives from the
approximation of the line integrals along lines which trace the
bounded domain \cite{MethodArchy}.

This report presents an extensive description of a vector field
recovery method inspired by \cite{MethodArchy}, elaborating on
fundamental assumptions and the ill conditioning of the problem and
defines the error bounds of the recovered solution. Such aspects are
critical for future implementations of the method in practical
applications like the inverse bioelectric field problem.

Moreover, the most interesting results from previous work on the
tomographic methods related to ray and Radon transform are
presented, including the basic theoretical foundation of the problem
and various practical considerations.

\section{Motivation}
In the present project, the final goal is the implementation of a
different approach for the EEG (Electroencephalography) analysis
employing the proposed vector field method. Rather than estimating
strengths or locations of the electric sources inside the brain,
which is a very complicated task, a reconstruction of the
corresponding static bioelectric field will be performed based on
the line integral measurements. This static bioelectric field can be
treated as an ``effective'' equivalent state of the brain at any
given instant. Thus for instance, health conditions and specific
pathologies (e.g. seizure disorder) may be recognized.

For this purpose, in the current report a robust mathematical and
physical model and subsequently an efficient numerical
implementation of the problem will be formulated. As a future stage,
this theoretical and numerical formulations will be adapted to the
real EEG problem with the help of experts and neuroscientists.

\section{Report Structure} The rest of the report consists of four
sections. The $2^{st}$ chapter is introductory and gives a brief
description of the previously related work as well as the
mathematical definition and theorems used by vector field
tomographic methods. The $3^{rd}$ chapter gives an overview of the
numerical vector field recovery method. Mathematical definitions,
physical assumptions and conditions for the recovery of an
irrotational vector field from line integral measurements without
considering any boundary condition are described. Also, the
approximation errors derived from the discretization of the line
integrals and the a-prior error bounds are estimated. In the
$4^{th}$ chapter, verification of the theoretical model by
performing simulations as well as the practical adaptation of this
model to the inverse bioelectric field problem using EEG 
measurements are described. Finally in the $5^{th}$ chapter, future
work is discussed.

\chapter{Previous Work}
In this chapter, we review several techniques that are important
mathematical prerequisites for a better understanding of the vector
field reconstruction problem. Moreover, an overview of the most
prominent potential applications is presented.

\section{Mathematical Preliminaries}

The basic mathematical tools for the vector field problem
formulation are presented in the following sections.

\subsection{Formulation of Vector Field Tomography problem}
The reconstruction of a scalar function $f(\textbf{n})$ from its
line integrals or projections in a bounded domain is a well known
problem. Today there are many practical and research applications in
different fields such as biomedicine (e.g. MRI, CT), acoustic and
seismic tomography which employ this method with great success and
accuracy. However, there is a variety of other applications, like
blood flow (velocity) in vessels or the diffusion tensor MRI problem
where the estimation and the imaging of a vector field can be
essential for the extraction of useful information. In these cases
tomographic vector methods intend to reconstruct these fields from
scalar measurements (projections) in a similar way to the scalar
tomographic methods.

The mathematical formulation of the vector tomographic problem is
given by the line integrals (projection measurements $I_L$)
\begin{equation}\label{eq:lineIntegralPreviousWork}
I_L=\int_L\textbf{F}\cdot
d\textbf{r}=\int\textbf{F}\cdot\hat{\textbf{r}}dr
\end{equation}
where $\hat{\textbf{r}}$ is the unit vector in the direction of line
L and $\textbf{F}$ is the vector field to be recovered.

A different formulation for the projection measurement in $2D$
domain commonly used in bibliography involves the 1-D Dirac delta
function such as
\begin{equation}
I_L=\int\int_{D}(F_x(x,y)\cos\phi+F_y(x,y)\sin\phi)\delta(x\sin\phi-y\cos\phi-\rho)dxdy
 \label{eq:VRT_2D}
 \end{equation}
where $x\sin\phi-y\cos\phi=\rho$ is the line function with
parameters $(\rho,\phi)$ defined as shown in figure
\ref{fig:LineParameters} and $D$ is the bounded domain where
$\textbf{F}(x,y)\neq0$.

\begin{figure}[!htb] \centering
\includegraphics[width=0.5\textwidth]{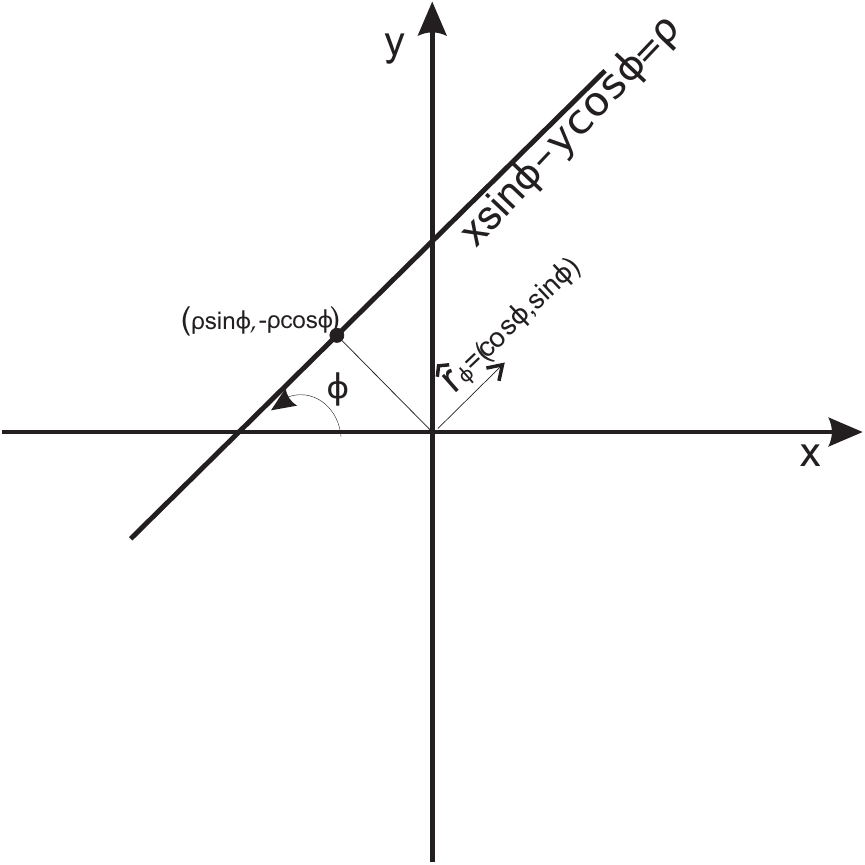} \caption{Line $L$ on $z$ plane with parameters $(\rho,\phi)$} where $0^0\leq\phi\leq180^0$ and $\rho\,\in\,\Re$.\label{fig:LineParameters}
\end{figure}

In general the vector tomographic problem is considered to be ill
posed since $\textbf{F}$ is defined by two or three components.
However, with the application of certain constraints, restrictions
and further assumptions there are ways to solve the problem.

\subsection{Helmholtz decomposition}
The Helmholtz decomposition \cite{MathForPhysics} is a fundamental
theorem of the vector calculus analysis as we shall see later.

It states that any vector $\textbf{F}$ which is twice continuously
differentiable and which, with its divergence and curl, vanishes
faster than $1/r^2$ at infinity, can be expressed uniquely as the
sum of a gradient and a curl as follows:
\[\textbf{F}=\textbf{F}_I+\textbf{F}_S \Rightarrow\]
\begin{equation}
\textbf{F}=-\nabla \Phi+ \nabla \times \textbf{A}
\label{eq:Decomposition}
\end{equation}
The scalar function $\Phi$ is called the scalar potential and
$\textbf{A}$ is the vector potential which should satisfy $\nabla
\cdot \textbf{A}=0$.

Since, $\nabla\times\textbf{F}_{I}=\nabla\times(\nabla\Phi)=0$,
component $\textbf{F}_{I}$ is called irrotational or curl-free while
$\textbf{F}_{S}$ is the solenoidal or divergence-free component as
it satisfies $\nabla\textbf{F}_{S}=\nabla\cdot(\nabla\times
\textbf{A})=0$.

In the case of a $2D$ vector field $\textbf{F}(x,y)$, the
decomposition equation becomes

\[\textbf{F}=-\nabla\Phi+\nabla\times A_{z}(x,y)\hat{z}\].

\subsection{Vectorial Ray Transform}
In tomographic theory, the line integral
\ref{eq:lineIntegralPreviousWork} is called ray transform. This
transform is closely related to the Radon transform
\cite{Book:RadonTransform} and coincides with it in two dimensions.
In higher dimensions, the ray transform of a function is defined by
integrating over lines rather than hyper-planes as the Radon
transform.

In particular, for a bounded volume $V$ ($\textbf{F}=0$ outside $V$)
and according to equation \ref{eq:lineIntegralPreviousWork}, the
vectorial ray transform can be expressed as
\begin{equation}\begin{split}
I(\phi,\theta,\textbf{p})=&\int_{L(\phi,\theta,\textbf{p})}\textbf{F}\cdot\hat{\textbf{r}}_{\phi,\theta}
dr=\\
=&\int_{L(\phi,\theta,\textbf{p})}F_x(x,y,z)\cos\phi\sin\theta
dr\\
+&\int_{L(\phi,\theta,\textbf{p})}F_y(x,y,z)\sin\phi\sin\theta
dr\\
+&\int_{L(\phi,\theta,\textbf{p})}F_z(x,y,z)\cos\theta dr
\label{eq:VRT}
\end{split}
\end{equation}
where $F_x$, $F_y$ and $F_z$ are the components of vector
$\textbf{F}$, $\phi$
and $\theta$
define the direction of the $\hat{\textbf{r}}_{\phi,\theta}$ unit
vector along line $L(\theta,\phi,\textbf{p})$ as shown in figure
\ref{fig:RayTransform} and $\textbf{p}=(x_p,y_p)$. Point
$\textbf{p}$ gives the coordinates of the line in the \emph{plane}
which passes through the origin of the axes and it is vertical to
$\hat{\textbf{r}}_{\phi,\theta}$ (see fig. \ref{fig:RayTransform}).

\begin{figure}[!htb] \centering
\includegraphics[width=0.7\textwidth]{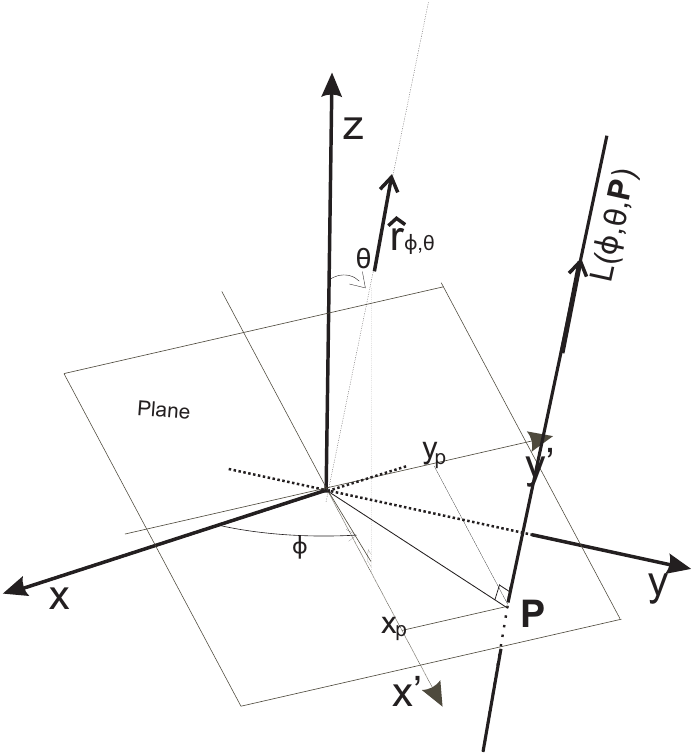} \caption{3D ray
transform geometry.}\label{fig:RayTransform}
\end{figure}

Consequently, the line integral \ref{eq:VRT} can be written as a
volume integral using the appropriate Dirac delta functions. Thus we
have
\begin{equation}\begin{split}
I(\phi,\theta,\textbf{p})&=\int\int\int_{V}F_x(x,y,z)\cos\phi\sin\theta \delta_{x_{p}}\delta_{y_{p}}dxdydz\\
&+\int\int\int_{V}F_y(x,y,z)\sin\phi\sin\theta\delta_{x_{p}}\delta_{y_{p}}dxdydz\\
&+\int\int\int_{V}F_z(x,y,z)\cos\theta\delta_{x_{p}}\delta_{y_{p}}dxdydz
 \label{eq:VRT1}\end{split}
 \end{equation}
 where \[\delta_{y_{p}}=\delta(x\sin\phi-y\cos\phi-y_{p})\] and \[\delta_{x_{p}}=\delta(-x\cos\phi\cos\theta-y\sin\phi\cos\theta+z\sin\theta-x_{p})\]

In two dimensions, $\theta=\pi/2$ and $p$ is the signed distance of
the line from the origin of the axes. Thus, equation \ref{eq:VRT1}
becomes
\begin{equation}
I(\phi,p)=\int\int_{D}(F_x(x,y)\cos\phi+F_y(x,y)\sin\phi)\delta(x\sin\phi-y\cos\phi-p)dxdy
 \label{eq:VRT_2D}
 \end{equation}
and $\textbf{F}(x,y)=0$ outside D.

\subsection{Central Slice Theorem}
The solution to the inverse \emph{scalar} ray transform is based on
the central slice theorem (CST). CST states that the values of the
$2D$ FT of scalar function $f(x,y)$ along a line with inclination
angle $\phi$ are given by the $1D$ FT of its projection $I(\phi,p)$.
This fact combined with the implementation of many practical
algorithms (e.g. back-projection) gave rise to the development of
accurate and robust reconstruction methods.

In the \emph{vectorial} ray transform, the CST does not help us
solve the problem. However, the formulation of the problem based on
the CST is important for better understanding the theoretical
approaches which will be described in the next section.

Let the Fourier Transform of $I(\phi,\theta,\textbf{p}=\{x_p,y_p\})$
be
\begin{equation} \widetilde{I}_{\phi,\theta}(k_1,k_2)=\int\int
I(\phi,\theta,x_p,y_p)e^{-i(k_1x_p+k_2y_p)}dx_pdy_p
\label{eq:CentalSliceTheorem1}
\end{equation}

Then according to equations \ref{eq:VRT} and \ref{eq:VRT1} we obtain
\begin{equation}
\widetilde{I}_{\phi,\theta}(k_1,k_2)=\cos\phi\sin\theta\widetilde{F}_x(u,v,w)+\sin\phi\sin\theta\widetilde{F}_y(u,v,w)+\cos\theta\widetilde{F}_z(u,v,w)
\label{eq:CentalSliceTheorem2}
\end{equation}
where $\widetilde{F}_x$, $\widetilde{F}_y$ and $\widetilde{F}_z$ are
the Fourier transforms of $F_x$, $F_y$ and $F_z$ respectively and
$u=k_1\sin\phi-k_2\cos\phi\cos\theta$,
$v=-k_1\cos\phi-k_2\sin\phi\cos\theta$ and  $w=k_2\sin\theta$.

Applying the Helmholtz decomposition (eq. \ref{eq:Decomposition}) we
have
\[F_x(x,y,z)=\frac{\partial {A}_z}{\partial y} -\frac{\partial A_y}{\partial z}-\frac{\partial \Phi}{\partial x}\]
\[F_y(x,y,z)=-\frac{\partial {A}_z}{\partial x} +\frac{\partial A_x}{\partial z}-\frac{\partial \Phi}{\partial y}\]
\[F_z(x,y,z)=\frac{\partial {A}_y}{\partial x} -\frac{\partial A_x}{\partial y}-\frac{\partial \Phi}{\partial z}\]

Their Fourier transform leads to
\[\widetilde{F}_x(u,v,w)=iv\widetilde{A}_z(u,v,w)-iw\widetilde{A}_y(u,v,w)-iu\widetilde{\Phi}(u,v,w)\]
\[\widetilde{F}_y(u,v,w)=iw\widetilde{A}_x(u,v,w)-iu\widetilde{A}_z(u,v,w)-iv\widetilde{\Phi}(u,v,w)\]
\[\widetilde{F}_z(u,v,w)=iu\widetilde{A}_y(u,v,w)-iv\widetilde{A}_x(u,v,w)-iw\widetilde{\Phi}(u,v,w)\]
when $\textbf{A}$ and $\Phi$ tend to zero on the volume boundaries.

Therefore the Fourier transform of the projection can be written as
\begin{equation}\label{eq:Fourier3D}\begin{split}
\widetilde{I}_{\phi,\theta}(k_1,k_2)=&i(k_1\cos\phi\cos\theta+k_2\sin\phi)\widetilde{A}_x\\+&i(
k_1\sin\phi\cos\theta-k_2\cos\phi)\widetilde{A}_y-ik_1\sin\theta\widetilde{A}_z
\end{split}\end{equation}
For the two dimensional case, $\theta=\pi/2$ and $k_2=0$ (as we have
only one variable)
\begin{equation}
\widetilde{I}_{\phi}(k)=ik\widetilde{A}_z(k\cos\phi,k\sin\phi)
\label{eq:Fourier2D}\end{equation} It is important to be mentioned
that irrotational component $\Phi$ vanishes in equations
\ref{eq:Fourier3D} and \ref{eq:Fourier2D}.

\section{Theoretical Approaches}\label{ch:RelW}
The most important theoretical and mathematical studies of the
vector field tomographic reconstruction from boundary measurements,
as well as the feasibility of this formulation to yield unique
solutions under certain constraints, were investigated only by a
small group of the research community working in this field. Norton
\cite{Norton1}, Baun and Haucks \cite{BraunAndHauck} and Prince
\cite{Prince94tomographicreconstruction}, \cite{PrinceOsman} gave a
step by step mathematical solution to the problem on bounded domains
employing the Radon transform theory. In the following subsection, a
description of their ideas and their methods is presented.

\subsection{Tomographic Vector Field Methods}
Norton in \cite{Norton1} and \cite{Norton2} was the first who
defined the full mathematical formulation of the two dimensional
problem. Norton proved that only the solenoidal component of a
vector field $\textbf{F}$ on a $D$ bounded domain can be uniquely
reconstructed from its line integrals. Moreover, he showed that when
the field $\textbf{F}$ is divergenceless i.e. there are no sources
or sinks in $D$, then both components can be recovered.

In particular, assuming a bounded vector field  $\textbf{F}$, i.e
$\textbf{F}=0$ outside a region $D$ which satisfies the homogeneous
Neumann conditions on $\partial{D}$ (on the field's boundaries), he
produced equation \ref{eq:Fourier2D} applying Helmzoltz
decomposition (eq. \ref{eq:Decomposition}) and the Central Slice
Theorem (eq.\ref{eq:CentalSliceTheorem1}). Therefore, he proved that
only the solenoidal component $\nabla\times\textbf{A}$ can be
determined.

Furthermore, in \cite{Norton1} he demonstrated that when vector
field $\textbf{F}$ is divergenceless ($\nabla\textbf{F}=0$), then
irrotational component $\Phi$ can be recovered.

From the divergence of the decomposition (eq.
\ref{eq:Decomposition}) we obtain
\[\nabla\textbf{F}=-\nabla(\nabla\Phi)+\nabla(\nabla\times\textbf{A})\Rightarrow \nabla^2\Phi=0\]

Thus, Norton was led to the Laplacian equation $\nabla^2\Phi=0$. The
solution of the Laplacian equation gives the irrotational component
and a full reconstruction of the field is possible. Norton employed
Green's theorem and  $\textbf{F}$'s boundary values for the
estimation of $\Phi$ on $D$.

Later Braun and Hauck \cite{BraunAndHauck} showed that the
projection of the orthogonal components of the vector function
(transverse projection measurement) leads to the reconstruction of
the irrotational component $\nabla\Phi$. So, they proposed that for
the full $2D$ field reconstruction, a longitudinal and a transverse
measurement are needed
\begin{equation*}
I_{\parallel}=\int\textbf{F}\cdot\hat{\textbf{r}}dr\;\;\mbox{and}\;\;
I_{\perp}=\int\textbf{F}\cdot\hat{\textbf{r}}_{\perp}dr
\end{equation*}
where $\hat{\textbf{r}}$ is the unit vector along the line and
$\hat{\textbf{r}}_{\perp}$ is the unit vector orthogonal to the
line.

Moreover, they examined the problem for non-homogeneous boundary
conditions. In that case, the irrotational and solenoidal
decomposition is not unique and they proposed to decompose the
vector into three components: the homogeneous irrotational, the
homogeneous solenoidal and the harmonic with its curl and divergence
being zero. They verified their method carrying out fluid flow
estimation experiments. With this method there is no need to assume
that there are no sources inside the domain. However, the difficulty
of taking transversal measurements as it was mention in
\cite{Norton2} makes the method quite impractical, especially for
Doppler back scattering methods.

Prince \cite{Prince94tomographicreconstruction} and Prince and Osman
\cite{PrinceOsman} extended the previous method to $3$ dimensions,
reconstructing both the solenoidal and the irrotational components
of $\textbf{F}$ from the inverse $3D$ Radon transform. Actually they
evolved the Braun-Hauck's method by defining a more general inner
product measurement which was called probe transform and it was
expressed as
\begin{equation}
G^{\textbf{p}}(\textbf{a},\rho)=\int_{\Re^3}\textbf{p}(\textbf{a},\rho)\cdot\textbf{F}(x,y,z)\delta(xa_x+ya_y+za_z-\rho)dxdydz
\label{eq:InnerProbeMeasurement}\end{equation} where $\textbf{p}$ is
the so-called  vector probe, $\rho$  the distance of the projection
plane from the origin and $\textbf{a}=(a_x,a_y,a_z)$ the normal
vector to the plane.

With the application of Helmholtz decomposition for homogeneous
field's boundaries (eq. \ref{eq:Decomposition}) and the Cental Slice
Theorem, equation \ref{eq:InnerProbeMeasurement} becomes

\[\widetilde{G}^{\textbf{p}}(\textbf{a},k)=(j2\pi k)\textbf{p(a)}\cdot[\Phi(k\textbf{a})\textbf{a}+\textbf{a}\times A(k\textbf{a})]\]
Therefore, if $\textbf{p}$ is orthogonal to $\textbf{a}$ then the
irrotational component is eliminated, while when $\textbf{p}$ is
parallel to $\textbf{a}$ then the solenoidal component vanishes. On
this basic principle Prince and Osman based their model for the
reconstruction of a $3D$ vector field.


\section{Proposed Applications}
A plethora of different applications have been proposed in the area
of vector field tomography. Some of the earlier studies by Johnson
\emph{et al.}\cite{Johnson1} and Johnson \cite{Johnson2} were
concerned with the reconstruction of the flow of a fluid by applying
numerical techniques (iterative algebraic reconstruction
techniques). Johnson \emph{et al.} \cite{Johnson1} used ultrasound
measurements (acoustic rays) to reconstruct the velocity field of
blood vessels. Later, other applications like optical polarization
tomography \cite{Hertz} for the estimation of electric field in a
Kerr material and oceanographic tomography \cite{OCEAN} were also
reported. In Kramar's thesis \cite{Kramar} a vector field method for
the estimation of the magnetic field of the sun's coronal is
presented, giving interesting results.

Moreover, in vector field literature there are many other proposed
applications in the area of Doppler back scattering, Optical
tomography, Photoelasticity and Nuclear Magnetic Resonance Plasma
physics \cite{VectorAndTensorTomography}. However, there are only a
few practical or commercial applications in this field.

\subsection{Vector Field Reconstruction and Biomedical Imaging}
Vector field tomography has not received much attention in the area
of medical applications. There are only a few papers
\cite{Doppler1,Doppler2,Doppler3,Doppler4,Duric,ConeBeam}
 and one PhD thesis  \cite{Javanovic} which present relevant methods.
The main area of research according to these papers are Doppler back
scattering for blood flow estimation, although till now there are
only simulations and theoretical formulations without performing any
real experiment or employing real data.

Moreover, the Lawrence Berkeley National Laboratory
\cite{VectorAndTensorTomography} has developed many tomographic
mathematical tools and algorithms for medical imaging issues. As it
is mention in \cite{VectorAndTensorTomography}, their work has
focused mainly on the implementation of algorithms for the
reconstruction of the $3D$ diffusion tensor field from MRI tensor
projections and iterative algorithms for solving the non-linear
diffusion tensor MRI problem.
\section{Summary} Study of previous work indicates that the
vector field tomography has practical potential. Till now much of
the work was devoted to the theoretical development and formulation
of the problem. Moreover it is clear that the acquisition of the
measurements and the performance of real experiments are quite
difficult tasks and a multidisciplinary collaboration is required.

\chapter{Vector Field Recovery Method: a Linear Inverse Problem}

In the previous chapter, an extended description of the mathematical
expression of the vector field tomographic problem in ray and Radon
transform sense was presented. A different approach for vector field
recovery method stemming from the numerical inverse problems theory
will be considered here.

The numerical solution of an inverse problem requires the definition
of a set of equations mathematically adapted to the physical
properties of the problem, subsequently, the design of the
geometrical model where these equations operate and finally the
discretization of the equations to form a numerical system such as
the approximated solution to be as close as possible to the real
solution of the model.

So, the vector field recovery problem can be considered as an
inverse problem which can be formulated as an operator equation of
the form
\[K\textbf{x}=y\] with $K$ being a linear operator between spaces $X$ and $Y$ over the field
$\Re$ and where the geometrical and numerical models are designed
according to the topological and error approximation requirements of
the problem.

The current method is based on the estimation of a vector field from
the line integrals (projection measurements) in an unbounded domain.
Thus, $K$ operator is an integral operator and the vector field
recovery method relies on solving a set of linear equations which
derives from a set of numerically approximated line integrals which
trace a bounded domain $\Omega$, and are expressed as
\[\int_L\textbf{E}\cdot d\textbf{r}=\Phi(\textbf{a})-\Phi(\textbf{b})\]
where $\textbf{E}$ is the irrotational vector field,
$d\textbf{r}=\hat{\textbf{r}}dr$ with $\hat{\textbf{r}}$ the unit
vector along line $L$ and $\Phi$ gives the boundary measurements at
starting point $\textbf{a}$ and endpoint $\textbf{b}$.

The initial idea was put forward in \cite{MethodArchy} where it was
shown that there is potential for a vector field to be recovered in
a finite number of points from boundary measurements. However, this
initial idea followed an intuitive approach to the problem as it
lacked the necessary conditions and assumptions about the recovered
field and the formulation of the equations. Moreover, there was no
clear and robust proof about the well or ill posedness of this
inverse problem.

Therefore, in the rest of this text:
\begin{itemize}
 \item the necessary preconditions and assumptions are defined such as the
mathematical formulation of the problem fits the physical vector
field properties as closely as possible;
\item an extended description of the vector field method
employing the line integral formulation is given;
\item the ill posedness of the vector field method in the
continuous domain is investigated and we show that the number of
independent equations which stem from the problem's formulation give
a final numerical system which is nearly rank deficient (ill
conditioned);
\item the
approximation and discretization errors resulted by the numerical
implementation of the problem are formulated. Finally, the solution
error of the numerical system is estimated and the conditions under
which this error is bounded are presented, revealing that the
discretization is a way of ``self-regularization'' of this ill
conditioned inverse problem.
\end{itemize}

\section{Preconditions and Assumptions}
Field $\textbf{E}$ is assumed bounded and continuous in a domain
$\Omega$, bandlimited, irrotational and quasi-static. 

In particular, the quasi-static condition implies that the field
behaves, at any instance, as if it is stationary. Moreover, the
field is considered irrotational and satisfies
$\nabla\times\textbf{E}=0$ and thus it can be represented by the
gradient of a scalar function $\Phi$. So, $\textbf{E}=-\nabla\Phi$
in a simply connected region (Poincare's Theorem
\cite{MathForPhysics}). Consequently, the gradient theorem gives
\begin{equation}\label{model equation1}
\int_c\textbf{E}\cdot d\textbf{c}=\int_c-\nabla\Phi
d\textbf{c}=\Phi(\textbf{a})-\Phi(\textbf{b})
\end{equation}
which will be the \emph{model equation} of our problem and implies
that the line integral of $\textbf{E}$ along any curve $c$ is
path-independent.

The irrotational property $\nabla\times{\textbf{E}}=0$ in integral
form can be expressed by applying the well-known Stokes' theorem
(curl-theorem) which relates the surface integral of the curl of a
vector field over a surface $S$ in Euclidean $3D$ space to the line
integral of the vector field over its boundary $\partial {S}$ such
as
\[\int_{S}\nabla\times\textbf{E}d\textbf{S}=\int_{\partial{S}}\textbf{E}\cdot d\textbf{c}\]
For an irrotational field obviously
\begin{equation}
\int_{\partial{S}}\textbf{E}\cdot d\textbf{c}=0
\label{eq:ClosedCurve}\end{equation} So, the path integral of
$\textbf{E}$ over a closed curve (path) is equal to zero.

Finally, vector field $\textbf{E}$ is band limited,
$\int\int_{|\Omega|} \|\textbf{E}\|dxdy< \infty$ and continuous in
$\Omega$ thus is can be expanded in Fourier Series as

\[ \textbf{E}(x,y)=\sum_{\textbf{k}}\textbf{e}_ke^{i\textbf{k}\langle x,y\rangle}+\bar{\textbf{e}}_ke^{-i\textbf{k}\langle x,y\rangle}\]
where $\textbf{k}=(n_x,n_y)$ with $n_x,n_y=0,\pm1,\pm2\dots$.

As the physical and mathematical properties of the vector field have
been defined, the mathematical and numerical formulation of the
problem will be explained in the following section.
\section{Methodology}
\subsection{Mathematical Modeling}
The formulation of the method is based on the idea in
\cite{Giannakidis2010} and \cite{MethodArchy} to approximately
reconstruct a vector field $\textbf{E}$ at a finite number of points
when a sufficiently large number of line integrals $I_{L_k}$ along
lines which trace the bounded domain, are known. The model equation
for the recovery of an irrotational field inside a bounded convex
domain $\Omega$  is given by
\begin{equation}\label{eq:VectorFieldIntegral}
I_{L_{k}}=\int_{L_{k}}\textbf{E}\cdot d\textbf{r}\end{equation}
where line $L_{k}$ traces the bounded domain and intersects it in
two points. As the field is irrotational, the line integral is
$I_{L_k}=\Phi(\textbf{a}_k)-\Phi(\textbf{b}_k)$ where $\Phi$ are the
boundary values at the intersection points $\textbf{a}_k$ and
$\textbf{b}_k$ with domain $\Omega$. So, a set of linear equations
\ref{eq:VectorFieldIntegral} can be acquired for a finite number of
known values $\Phi$ on the boundaries of $\Omega$.

\subsection{Geometric Model}\label{subsection:Geometric Model}
As the model equations have been defined, the next step is the
geometric model generation. Generally, the geometric model is a
discrete domain of specific shape where the model equations are
valid. For instance, if the line integral equations (\emph{model
equations}) were employed for field recovery from scalp potential
recordings (e.g. EEG), then the geometric model would be a $3D$ mesh
with scalp shape.

In our initial approach for the evaluation of the method, a simple
geometric model was designed as described in \cite{MethodArchy}. In
particular, a discrete version of a $2D$ continuous square domain
$\Omega=:\{(x,y)\in[-U,U]^2\}$ (fig.\ref{fig:Discrete field}) was
defined using elements of constant size $[P\times P]$ called cells
or tiles and $N=U/P$.
\begin{figure}[!htb] \centering
\includegraphics[width=0.5\textwidth]{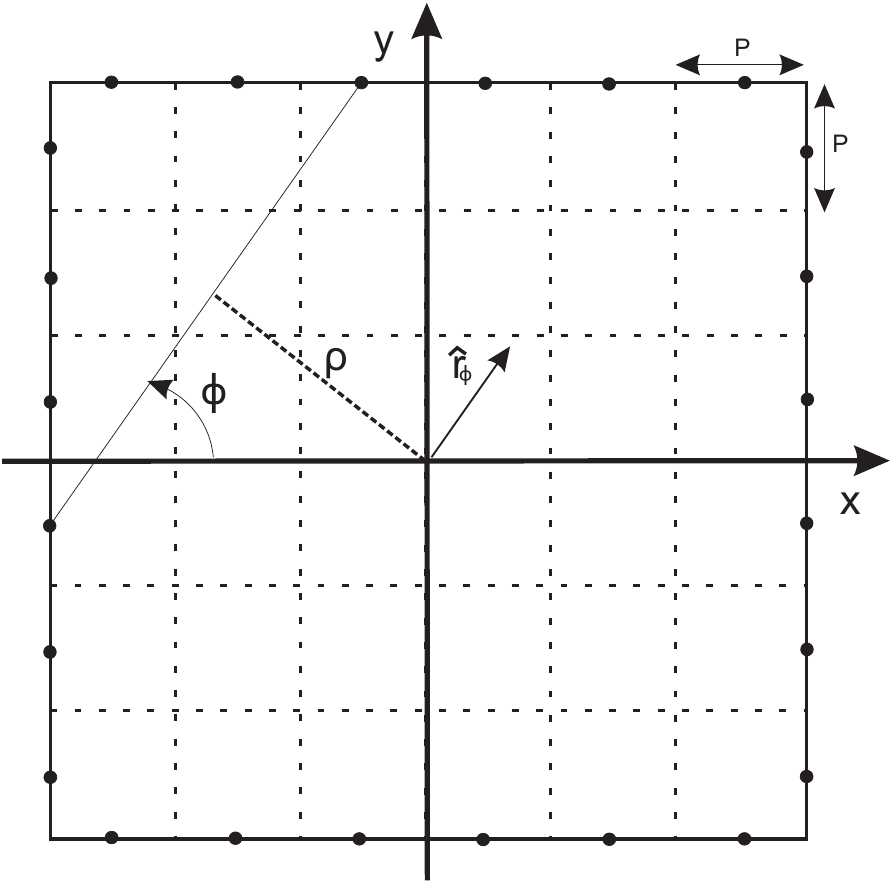} \caption{A line $L(\phi,\rho)$ of integration on a square discrete domain
where $\phi$ is the counterclockwise angle between the positive $x$
axis and the line and $\rho$ the signed distance of the line from
the origin.}\label{fig:Discrete field}
\end{figure}
The goal was to recover the field in each cell solving a numerical
system of a discretized version of \emph{model equations}
\ref{eq:VectorFieldIntegral}. The number and the positions of the
tracing lines $L_k$ in the bounded domain were defined by pairs of
sensors (boundary measurements) which were placed in the middle of
the boundary edges of all boundary cell (fig.\ref{fig:Discrete
field}). Each tracing line connected a pair of sensors which did not
belong to the same side of the square domain and thus for a $N=U/P$
number of edge cells or $N$ sensors in each side of the domain, the
connected pairs led to $6N^2$ \emph{model equations} (line
integrals). The domain had $N\times N$ cells and the $2D$ vector
field had two components. Thus the number of unknowns (value of each
cell) was $2[N\times N]$ and the number of equations was $6N^2$ .

\subsection{Numerical
Implementation}\label{subsection:NumericalImplementation} For the
numerical implementation of the method, the line integral
\ref{eq:VectorFieldIntegral} was approximated using the Riemann's
sum. Thus, \begin{equation} I_{L_k}\approx\sum_{(x_m,y_m)\in
L_k}\textbf{E}(x_m,y_m)\cdot \hat{\textbf{r}}_{\phi}\Delta r
\label{eq:Summation}\end{equation} where $\textbf{E}(x_m,y_m)$ are
the unknown vector values at sampling points $(x_m,y_m)$ along line
$L(\phi,\rho)=\{(x,y)|\,x\sin\phi-y\cos\phi=\rho\}$ with
$0^0\leq\phi\leq180^0$ and $\rho \in \Re$ and
$\hat{\textbf{r}}_{\phi}\Delta r=(\cos
\phi\hat{\textbf{x}}+\sin\phi\hat{\textbf{y}})\Delta r$
(fig.\ref{fig:Discrete field}).

For the numerical approximation, the samples $\textbf{E}(x_m,y_m)$
are assigned to values $\textbf{E}[i,j]=(e_x[i,j],e_y[i,j])$ based
on an interpolation scheme. The simple case of the nearest neighbor
approximation in a $2D$ square domain $[-U,U]^2$ leads to
$i=\lceil\frac{x_m+U}{P}\rceil$ and
$j=\lceil{\frac{y_m+U}{P}}\rceil$.

Finally, all the approximated line integrals give a set of algebraic
equations which can be represented by
\begin{equation}
\textbf{b}=\textbf{A}\textbf{x}
\label{eq:SystemOfLinearEquations1}\end{equation}
where $\textbf{x}$ contains the unknown vector field components in
finite points with index $u$, $\textbf{b}$ is the column vector with
the measured values of the line integrals $I_{L_{k}}$. The elements
of the transfer matrix $\textbf{A}$, $a_{k,u}$ represent the weight
of projection of the unknown element $u$ on the $I_{L_{k}}$. For the
case of a discrete square domain as it was described in subsection
\ref{subsection:Geometric Model}, the $4N$ sensors around the
boundaries give $6N^2$ equations and as the number of the unknowns
is $2N^2$, at first sight we conclude that we deal with an
over-determined system.

\section{Ill Posedness and Ill Conditioning of the Inverse problem}\label{subsection:IllPosednessConsideration}
If we ignore the physical properties of the field, the
reconstruction formulation described in the subsections
\ref{subsection:Geometric Model} and
\ref{subsection:NumericalImplementation} seems correct and robust as
the system \ref{eq:SystemOfLinearEquations1} of algebraic equations
is over-determined and thus with the least square method, the
problem can be solved. However, the majority of the inverse problems
are typically ill-posed.
\begin{figure}[!htb] \centering
\includegraphics[width=0.5\textwidth]{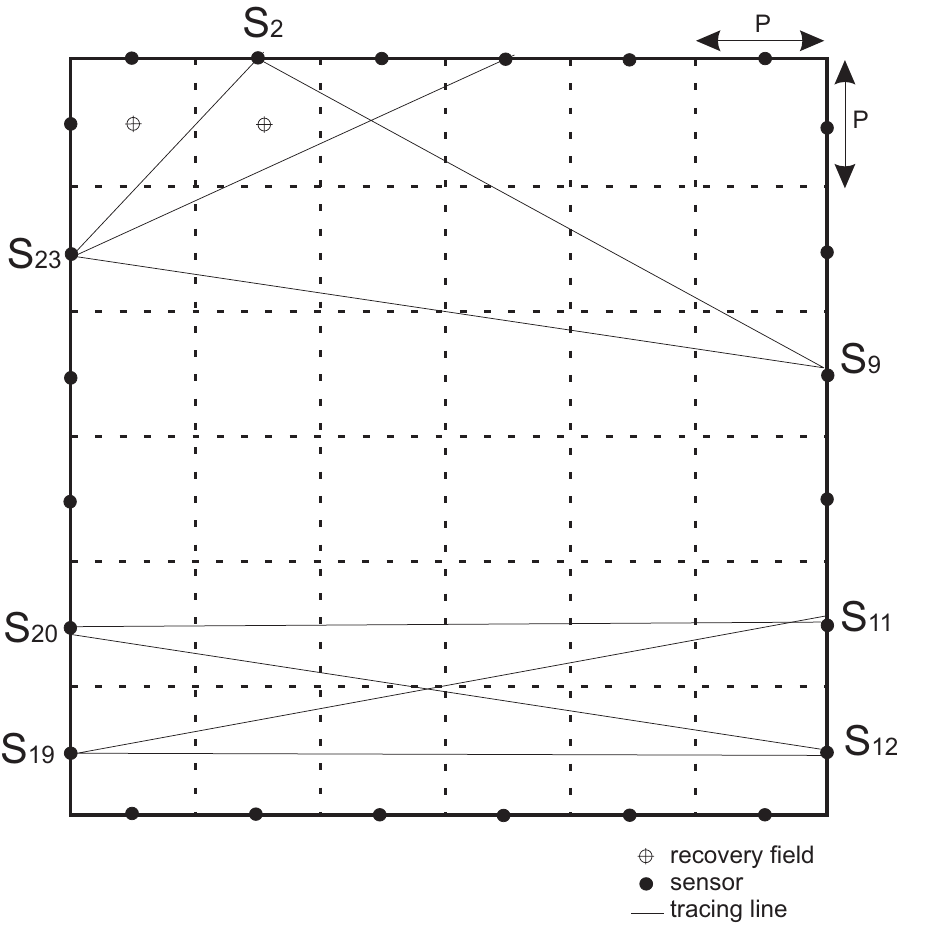} \caption{Tracing lines in the discrete domain which create closed-paths}\label{fig:Closed Paths}
\end{figure}

In the current problem, the irrotational assumption of the field
implies that a line integral along a closed path is zero. More
precisely, taking all pairs of sensors which do not belong to the
same side we obtain tracing lines which create closed paths (fig.
\ref{fig:Closed Paths}). Particulary, when the tracing lines $L_1,
L_2 \dots$ form a closed curve, according to equations
\ref{eq:ClosedCurve} and \ref{eq:VectorFieldIntegral} we obtain
\begin{equation}\label{eq:LinearDependenIntegrals}
I_{L_{1}}+I_{L_{2}}+\dots+I_{L_{n}}=0
\end{equation}
or
\[ -I_{L_{1}}=I_{L_{2}}+\dots+I_{L_{n}}\] This indicates that a line
integral can be expressed as a linear combination of other line
integrals and thus we have linearly dependent equations $I_{L_{k}}$.
For instance, in figure \ref{fig:Closed Paths} the lines which
connect sensors S$_2$, S$_9$ and S$_{23}$ create a closed loop and
thus for equations (line integrals) $I_{L_{2-9}}$, $I_{L_{9-23}}$
and $I_{L_{23-2}}$ we have
\begin{align*}
I_{L_{2-9}}+I_{L_{9-23}}+I_{L_{23-2}}&=\\
=\int_{{L_{2-9}}}\textbf{E}\cdot d\textbf{r}+
\int_{{L_{9-23}}}\textbf{E}\cdot d\textbf{r}+&\int_{{L_{23-2}}}\textbf{E} d\textbf{r}=\\
=\int_{{L_{2-2}}}\textbf{E}\cdot d\textbf{r}&=0
\end{align*}
Therefore, only two equations can be assumed independent since any
one of the three can be expressed as a linear combination of the
other two equations.

The linear dependencies are quite significant in the continuous
domain. On the other hand, in the discrete domain, where the
integral is approximated by a summation, the dependencies are not so
obvious as the accuracy of the continuous domain is lacking. Thus,
taking Riemann's summation (eq.\ref{eq:Summation}), leads to sums of
equations close to zero
\[I_{L_{1}}+I_{L_{2}}+\dots+I_{L_{n}}\approx 0\] If the accuracy
improves, i.e rather than applying nearest neighbor approximation, a
different interpolation scheme like bilinear, cubic or more
sophisticated techniques e.g. finite elements and a grid refinement
of the bounded domain, can lead to more obvious equation
dependencies.

One important task is the examination of the stability of the
solution of linear system \ref{eq:SystemOfLinearEquations1} i.e. to
check whether the independent equations are enough to give a unique
solution and whether matrix $\textbf{A}$ has full rank.

\subsubsection{Ill Conditioning: Number of Independent Equations}
The number of independent equations is important for the solution of
the problem since in the case that this number is less than the
unknowns, the system is under-determined and different mathematical
tools are needed.

In order to define the number of independent equation, we have to
exclude tracing lines which ``close'' paths such as making sure that
any line starting from one sensor does not end up to the same
sensor. The number of independent equations can be defined based on
the fundamental properties of graph theory \cite{GraphBook}.

More specifically, we assume that the $4N$ sensors along the
boundary of the square domain are the ``vertices'' of a graph G and
that the lines which connect two sensors are the G graph's
``edges''. The main property that this graph should satisfy is that
any two ``vertices'' are linked by a unique path or in other words
that the graph should be connected and without cycles. According to
graph theory, a ``tree'' is a undirected simple graph G that
satisfies the previous condition.

\begin{figure}[!htb] \centering
\includegraphics[width=0.1\textwidth]{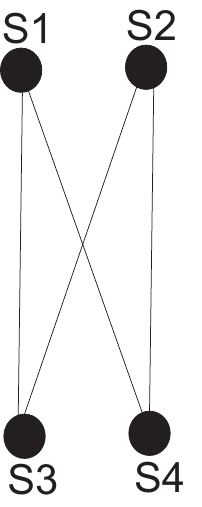}
\caption{Graph G which is not acyclic. If we extract one edge then
it becomes a tree.}\label{fig:Graph}
\end{figure}

Moreover, it is known that a connected, undirected, acyclic graph
with  $N$ ``vertices'' has $N-1$ ``edges''. Therefore in \textbf{a
$2D$ square domain (subsection \ref{subsection:Geometric Model})
with $4N$ sensors along the boundaries, the maximum number of
tracing lines in order to avoid loops is $4N-1$ and so the number of
independent equations (independent line integrals) is $4(N-1)$. In
all, the system of equations may have at most $4N-1$ independent
equations. Obviously, for a system with $2N^2$ unknowns, the
$4(N-1)$ equations lead to an under-determined case.}

As a result, the transfer matrix $\textbf{A}$ of system
\ref{eq:SystemOfLinearEquations1} approximates a rank-deficient
matrix i.e. there are nearly linearly dependent lines and the linear
system may be inconsistent and severely ill conditioned.

\subsection{Ill Conditioning Indicators}
Generally speaking, the ill posedness technically applies to
continuous problems. The discrete version of an ill posed problem
may or may not be severely ill conditioned.

The discrete approximation will behave similarly to the continuous
case as the accuracy of the approximation increases. With a
``rough'' approximation scheme and coarse discretization of the
bounded domain, the linear dependencies of the equations are
eliminated and transfer matrix $\textbf{A}$ of system
\ref{eq:SystemOfLinearEquations1} is not rank deficient.

The condition number of transfer matrix $\textbf{A}$ (system
\ref{eq:SystemOfLinearEquations1}) and the magnitude of the singular
values of $\textbf{A}$ are reliable indicators of how close to rank
deficiency and consequently to ill conditioning the numerical system
is. For the case where a bounded domain
$\Omega:\{(x,y)\rightarrow[-5.5,5.5]^2\subset \Re^2\}$ is
discretized employing cells of size $P\times P=1\times 1$, and there
are $4\times U/P= 4\times N=121$ sensors along the boundaries (see
subsection \ref{subsection:Geometric Model}), $6N^2=726$ equations
and the field is created by a single charge in position $(-19,19)$
on the \emph{z}-plane as in \cite{MethodArchy}, the condition number
is $84$, which is not so large in order to deal with a severely
ill-conditioned system. Moreover, the simulation results, which will
be presented in the next chapter, show that the estimated solution
is not far from the real one.

So, under certain conditions, system
\ref{eq:SystemOfLinearEquations1} which was derived from the
numerical approximation of the line-integral can give acceptable
results and the discretization process can be assumed as a kind of
``self-regularization'' (regularization by discretization) of the
continuous ill posed problem.

The aim of the work presented next is the mathematical definition of
the numerical errors due to the line-integrals approximations and
how these errors are related with the ill conditioning of the linear
system \ref{eq:SystemOfLinearEquations1} defining an upper error
bound of the system's numerical solution.
\section{Approximation Errors}
For the numerical solution of a line integral system, one has to
discretize the continuous problem and reduce it to a finite system
of linear equations. The discretization process introduces
approximation and rounding errors to the set of line integrals
(\emph{model equations} \ref{eq:VectorFieldIntegral}) which can be
interpreted as perturbations. First, the mathematical formulation
and nature of these errors will be described in the following
analysis. Subsequently, the error of the numerical solution will be
estimated.

\subsection{ Description and Derivation of the Approximation Errors}
The estimation of the line integral
(eq.\ref{eq:VectorFieldIntegral}) in a $2D$ domain is given by
Riemmann's integral
\begin{equation}\label{eq:RiemannSumApproximation}
I_{L_{\rho,\phi}}=\lim_{\max\{\Delta
r_k\}\rightarrow0}\sum_{k}\textbf{E}(x_k,y_k)\cdot
\hat{\textbf{r}}_{\phi}\Delta r_k
\end{equation}
where $(x_k,y_k)$ are the coordinates of the samples
$\textbf{E}(x_k,y_k)$ of the field along line
$L_{\rho,\phi}=\{(x,y)|x\sin\phi-y\cos\phi=\rho\}$ $\in\Omega$ with
$\hat{\textbf{r}}_{\phi}=(\cos\phi,\sin\phi)$ the unit vector in the
direction of the line (fig.\ref{fig:Discrete field}) and $\Delta
r_k$ the sampling step.

The coordinates of the samples are $x_k=x_{k-1}+\Delta r_k \cos\phi$
and $y_k=y_{k-1}+\Delta r_k \sin\phi$ such as $(x_k,y_k)\in\Omega$.

Substituting $\textbf{E}$ with its vector components $(e_x,e_y)$,
equation \ref{eq:RiemannSumApproximation} is expressed as
\begin{equation}\label{eq:RiemannSumApproximationDecomposed}
I_{L_{\rho,\phi}} = \lim_{\max\{\Delta
r_k\}\rightarrow0}\sum_{k}\{e_x(x_k,y_k)\cos\phi\Delta r_k +
e_y(x_k,y_k)\sin\phi\Delta r_k\}
\end{equation}
The numerical treatment of the problem is based on the assignment of
each sample $\textbf{E}(x_k,y_k)$ to a value $\textbf{E}[i,j]$ of an
element $[i,j]$ of the discrete domain. This assignment is actually
a quantization process or ``mapping'' of the vector field samples to
a finite set of possible discrete values
(fig.\ref{fig:Quantization}).

\begin{figure}[!htb] \centering
\includegraphics[width=0.4\textwidth]{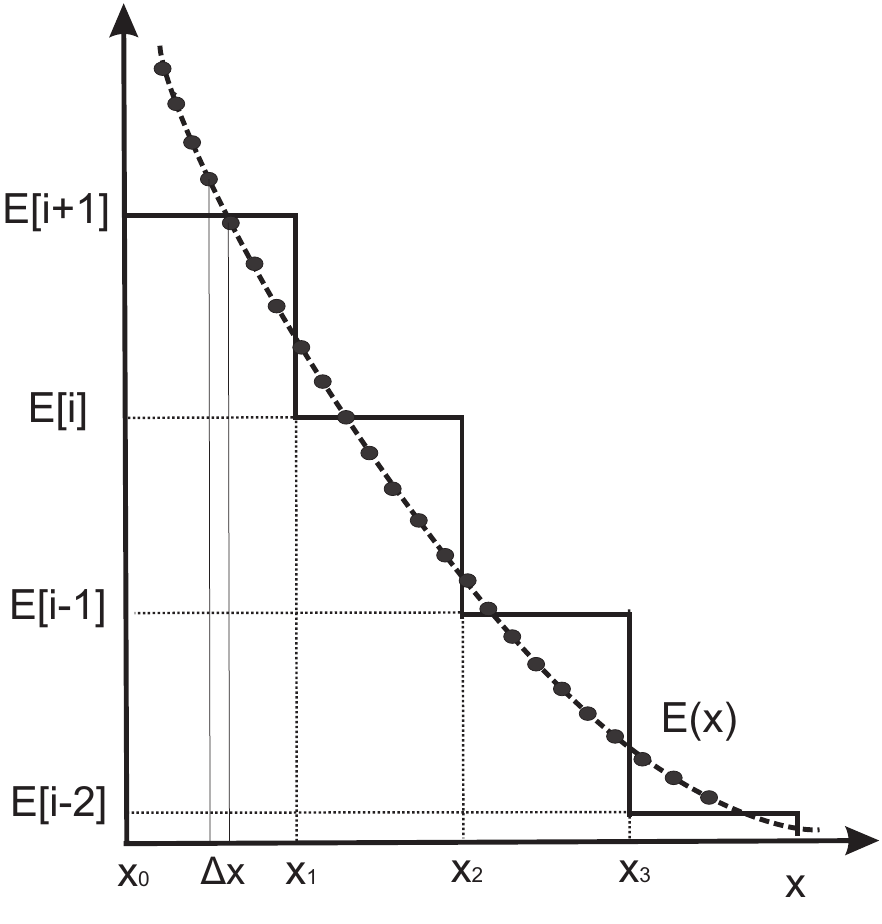} \caption{1D equivalent representation for the assignment of the samples of field $\textbf{E}$ to
the middle cell values (quantization process). In the current
representation, the samples of $E(x)$ are assigned to $E[i]$ values,
where $\Delta x$ is the sampling step.}\label{fig:Quantization}
\end{figure}

So, each sample $\textbf{E}(x_k,y_k)$ is the sum of
$\textbf{E}[i,j]$ plus an error vector $\delta\textbf{E}(x_k,y_k)$.
This can be expressed as
\begin{equation}\label{eq:NearestNeigbor}
e_x(x_k,y_k)=e_x[i,j]+ \delta e_x(x_k,y_k)\end{equation} or briefly
as
\[e_x^k=e_x[i,j]+\delta e_x^k\]
Similarly, the \textit{y}-component is
\[e_y^k=e_y[i,j]+\delta e_y^k\]
Let us assume that there is a line segment $\Delta L_{ij}\subseteq
L_{\rho,\phi}$ with length $\|\Delta L_{ij}\|$
 which lies on element(cell) $[i,j]$ as in figure \ref{fig:Discrete field Line
 Segment}.
Moreover, considering the sampling step to be constant $\Delta
r_k=\Delta r$ and $N_{ij}$ the number of samples
$\textbf{E}^k=(e_x^k,e_y^k)$ in $\Delta L_{ij}$, Riemann's sum along
segment $\Delta L_{ij}$ is
\begin{equation}\label{eq:groupingTheSampleOfCell}
\begin{split}
I_{\Delta L_{ij},\rho,\phi}= &\sum_{(e_x^k,e_y^k)\in\Delta
L_{ij}}\{e_x^k\cos \phi\Delta r+e_y^k\sin \phi\Delta r\}=\\=&
\sum_{(e_x^k,e_y^k)\in\Delta L_{ij}}\left\{(e_x[i,j]+\delta
e_x^k)\cos \phi\Delta r+\left(e_y[i,j]+\delta e_y^k\right)\sin
\phi\Delta r\right\}=
\\
=&\left(N_{ij}e_x[i,j]+\sum_{k}\delta e_x^k\right)\cos\phi\Delta r
+\left(N_{ij}e_y[i,j] +\sum_{k}\delta e_y^k\right)\sin\phi\Delta r
\end{split}
\end{equation}

Where $N_{ij}\in\aleph$ is the number of samples of $\Delta
L_{ij}\subseteq L_{\rho,\phi}$ and segment $\Delta L_{ij}$ lies
inside element $[i,j]$ of the discrete domain.

\begin{figure}[!htb] \centering
\includegraphics[width=0.5\textwidth]{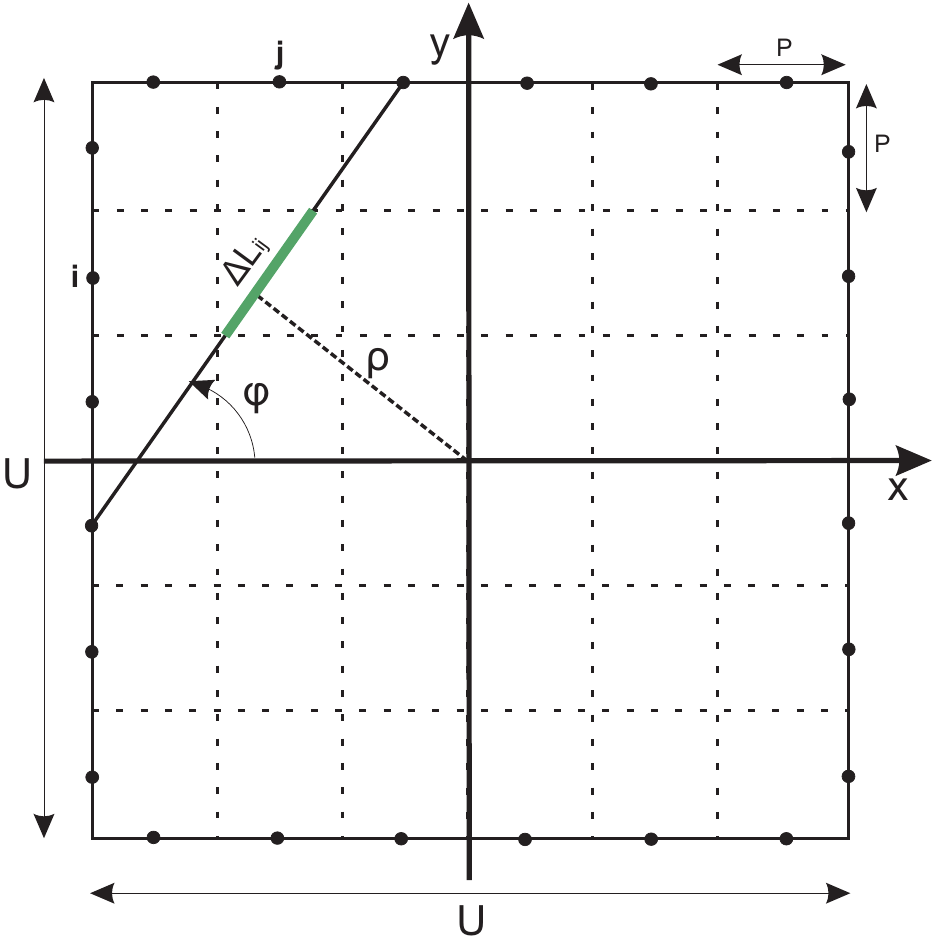} \caption{Segment $\Delta L_{ij}$ of line $L$ which
lies inside element $[i,j]$ in a $2D$ square
domain.}\label{fig:Discrete field Line Segment}
\end{figure}

The error vector caused by the discretization process in an element
$[i,j]$ is defined according to
\begin{equation}\label{eq:errorComponents}(\delta \tilde{e}_x[i,j],\delta \tilde{e}_y[i,j])=
\left(\sum_{k=N_1}^{N_2}\delta e_x^k ,\sum_{k=N_1}^{N_2}\delta
e_y^k\right)\end{equation} where $N_2-N_1+1=N_{ij}$ is the number of
samples along $\|\Delta L_{ij}\|$.

The line integral (eq.\ref{eq:RiemannSumApproximationDecomposed})
can be represented as a sum of line integrals
(eq.\ref{eq:groupingTheSampleOfCell}) along segments $\Delta L_{ij}\subseteq L_{\rho,\phi}$.\\
Thus,
\begin{equation}\label{eq:lineIntegralWithErrorsFormulations}
\begin{split}
I_{L_{\rho,\phi}}=&\lim_{\Delta r\rightarrow0}\sum_{\Delta
L_{ij}}I_{\Delta L_{ij},\rho,\phi}=\\=&\lim_{\Delta
r\rightarrow0}\sum_{\Delta L_{ij}}\sum_{\underset{\in\,\,\Delta
L_{i,j}}{(e_x^k,e_y^k)}}\{e_x^k\cos \phi\Delta r+e_y^k\sin
\phi\Delta r\}=\\=& \lim_{\Delta r\rightarrow0}\sum_{\Delta
L_{ij}}\left\{\left(N_{ij}e_x[i,j]+\delta
\tilde{e}_x[i,j]\right)\cos\phi\Delta r +
\left(N_{ij}e_y[i,j]+\delta\tilde{e}_y[i,j]\right)\sin\phi\Delta
r\right\}
\end{split}
\end{equation}
For simplicity, in the following equations indices $i$ and $j$ are
substituted by a single index $m$.

So equation \ref{eq:lineIntegralWithErrorsFormulations} becomes
\begin{equation}\label{eq:lineIntegralWithErrorsFormulations_M}
I_{L_{\rho,\phi}}=\lim_{\Delta r\rightarrow0}\sum_{\Delta
L_m}\left\{(N_m e_x[m]+\delta \tilde{e}_x[m])\cos\phi\Delta r +
(N_{m}e_y[m]+\delta\tilde{e}_y[m])\sin\phi\Delta r\right\}
\end{equation}

\begin{figure}[!htb]
\centering
\includegraphics[width=0.7\textwidth]{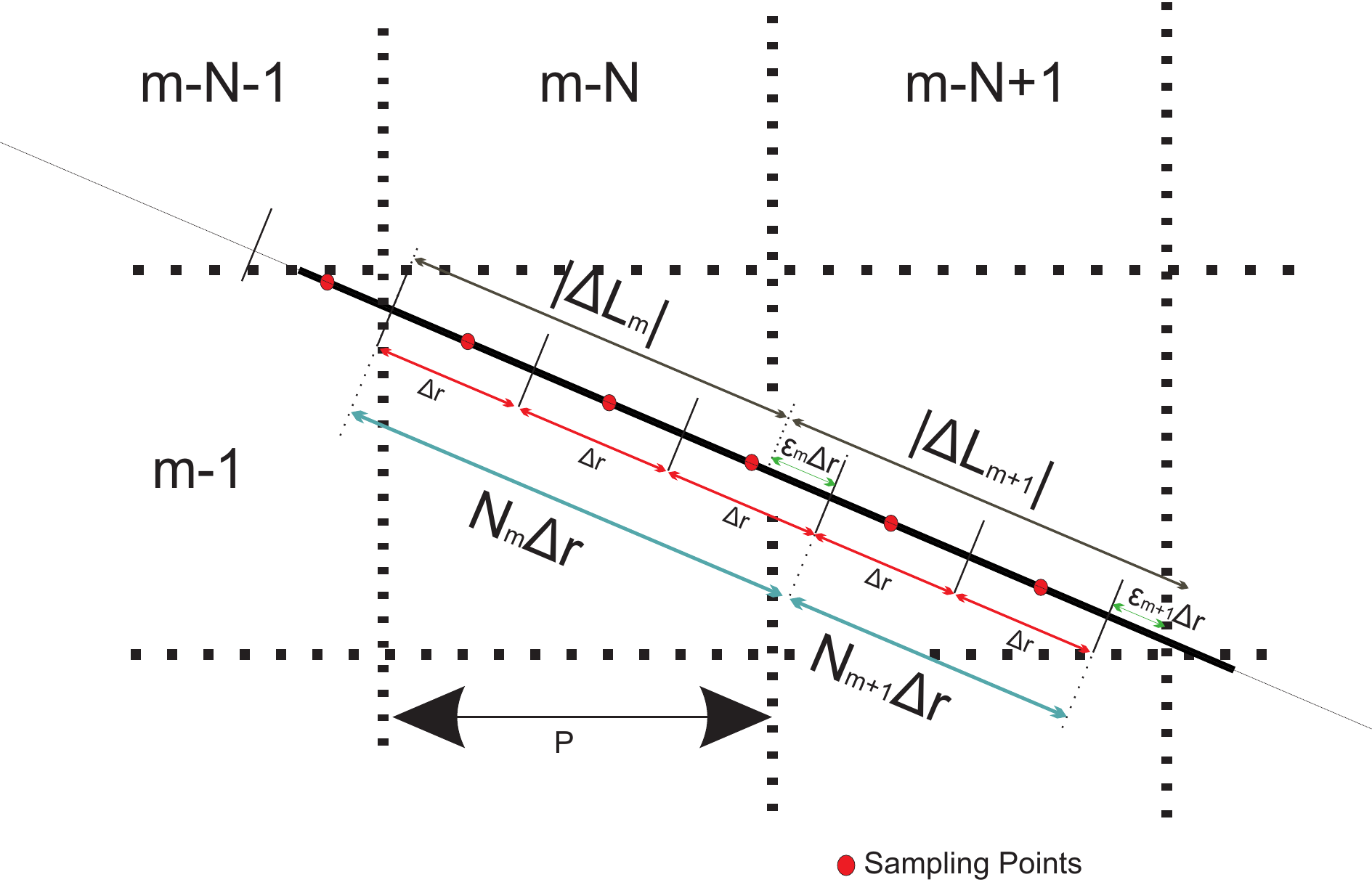} \caption{The sampling error
$\varepsilon_m$ in element $[m]$ is less than zero as  it was
defined by $N_m\Delta r=\|\Delta L_m\|-\varepsilon_m\Delta r$ while
in element $[m+1]$, $\varepsilon_{m+1}>0$.}\label{fig:SamplingError}
\end{figure}
The quantization error term due to the discretization process is
defined as
\begin{equation}\label{eq:discretizationError}\delta\tilde{E}=\lim_{\Delta
r\rightarrow0}\sum_{\Delta L_m}\left(\delta
\tilde{e}_x[m]\cos\phi\Delta r +\delta \tilde{e}_y[m]\sin\phi\Delta
r\right)\end{equation} Thus, equation
\ref{eq:lineIntegralWithErrorsFormulations_M} becomes
\begin{equation}\label{eq:lineIntegralWithErrorsFormulations_M2}
I_{L_{\rho,\phi}}=\lim_{\Delta r\rightarrow0}\sum_{\Delta
L_m}\left\{N_m \Delta
r(e_x[m]\cos\phi+e_y[m]\sin\phi)\right\}+\delta\tilde{E}
\end{equation}
Sampling step $\Delta r$ is finite and $\min\{\|\Delta
L_m\|\}\gg\Delta r>0$. Term $N_m\Delta r$ represents the length of a
segment of line $L_{\rho,\phi}$ which ``corresponds'' to element
$[m]$, where $N_m$ is the number of samples in element $[m]$
(fig.\ref{fig:SamplingError}). As the sampling step is $\Delta
r\gg0$, $N_m\Delta r\neq\|\Delta L_m\|$ as shown in figure
\ref{fig:SamplingError}. Thus, $N_m\Delta r=\|\Delta
L_m\|-\varepsilon_m\Delta r$ with $|\varepsilon_m|<1$ being a
sampling error coefficient.

The length of line $L_{\rho,\phi}$ is
\[\|L_{\rho,\phi}\|=\sum_{m} N_m\Delta r=\sum_{m}(\|\Delta
L_m\|-\varepsilon_m\Delta r)=\sum_{m}\|\Delta L_m\|
\,\,\,\mbox{and}\,\,\, \sum_m\varepsilon_m\Delta r=0.\]

Therefore, equation \ref{eq:lineIntegralWithErrorsFormulations_M2}
is equal to
\begin{equation*}
\begin{split}
I_{L_{\rho,\phi}}=& \lim_{\Delta r\rightarrow0}\sum_{\Delta L_{m}}
\left\{(\|\Delta L_m\|-\varepsilon_m\Delta
r)(e_x[m]\cos\phi+e_y[m]\sin\phi)\right\}+\delta\tilde{E}=\\=&
\sum_{\Delta L_{m}} \left\{\|\Delta
L_m\|(e_x[m]\cos\phi+e_y[m]\sin\phi)\right\}+\delta\tilde{E}
\end{split}
\end{equation*}
Finally, the line integral $I_{L_{\rho,\phi}}$ can be expressed as
the sum of two terms, one which is the discretization process term
$\sum_{\Delta L_{m}} \left\{\|\Delta
L_m\|(e_x[m]\cos\phi+e_y[m]\sin\phi)\right\}$ and another which is
error term $\delta \tilde{E}$ resulted by the discretization.

So,
\begin{equation}
\label{eq:lineIntegralWithErrorsFormulationsFinal}
I_{L_{\rho,\phi}}=\sum_{\Delta L_{m}} \left\{\|\Delta
L_m\|(e_x[m]\cos\phi+e_y[m]\sin\phi)\right\}+\delta\tilde{E}
\end{equation}
For the computational estimation of the line integrals, a finite
number $N$ of samples $\textbf{E}(x_k,y_k)$ along line
$L_{\rho,\phi}$ are assigned to the discrete values $\textbf{E}[m]$
based on an interpolation scheme (e.g. ``rough'' nearest neighbor).
Thus, the approximated line integral equation has the form
\begin{equation}\label{eq:ApproximationOflineIntegral}
\begin{split}
\tilde{I}_{L_{\rho,\phi}}=&\sum_{\Delta L_{m}}\left\{N_{m}\Delta
r(e_x[m]\cos\phi +e_y[m]\sin\phi)\right\}=\\=& \sum_{\Delta
L_{m}}\{(\|\Delta L_{m}\|-\varepsilon_{m}\Delta r) (e_x[m]\cos\phi +
e_y[m]\sin\phi)\}=\\=& \sum_{\Delta L_{m}}\{\|\Delta L_{m}\|
(e_x[m]\cos\phi + e_y[m]\sin\phi)\}-\varepsilon(\Delta r)
\end{split}
\end{equation}
Where
\begin{equation}\label{eq:samplingErrorFormula}
\varepsilon(\Delta r)=\sum_{\Delta L_{m}}\varepsilon_{m}\Delta
r\{e_x[m]\cos\phi + e_y[m]\sin\phi\}
\end{equation} is the
sampling error. Taking the difference between $I_{\rho,\phi}$
(eq.\ref{eq:lineIntegralWithErrorsFormulationsFinal}) and
$\tilde{I}_{\rho,\phi}$ (eq.\ref{eq:ApproximationOflineIntegral})
\begin{equation}\boxed{\label{eq:FinalErrors_Difference betweenIntegrals}
I_{\rho,\phi}-\tilde{I}_{\rho,\phi}=\delta\tilde{E}+\varepsilon(\Delta
r)}
\end{equation}
There are two different types of error, $\delta\tilde{E}$ which is
resulted by the discretization process and $\varepsilon(\Delta r)$
which is caused by sampling step $\Delta r>0$ along the integral
line. Obviously, as the sampling step $\Delta r\rightarrow 0$,
$\varepsilon(\Delta r)$ is eliminated.

\subsection{A Priori Error Estimate}
The numerical estimation of the field is based on the solution of
the linear system \ref{eq:SystemOfLinearEquations1} (subsection
\ref{subsection:NumericalImplementation}) in a Least Square (LS)
sense employing the numerical approximation of the line integrals.

For a grid $N\times N$ with $4N$ known boundary values, the linear
system has $6N^2$ equations and the unknown vector components are
$2N^2$. Hence, the LS system is expressed as
\begin{equation}\label{eq:LSLinearSystem}
\textbf{b}=\bar{\textbf{A}}\bar{\textbf{x}}_{LS}\end{equation}
\begin{center}
\end{center}
\begin{equation}\label{system:PerturbedSystem}
\left[ \begin{array}{c} I_{L_{\rho_1,\phi_1}} \\ I_{L_{\rho_2,\phi_2}}\\
\dots\\I_{L_{\rho_p,\phi_p}}
\end{array} \right] =
\begin{bmatrix}
\bar{a}_{11} & \bar{a}_{12}&\dots &\bar{a}_{1l/2+1}&\dots& \bar{a}_{1l}\\
 \bar{a}_{21}& \bar{a}_{22}&\dots & \dots &\dots & \dots\\
 \dots &\dots &\dots & \dots & \dots &\dots\\
 \dots & \dots&\dots & \dots& \dots&\bar{a}_{pl}
\end{bmatrix} \left[
\begin{array}{c} e_{x_{LS}}[1] \\ e_{x_{LS}}[2] \\ \vdots \\e_{y_{LS}}[1]\\ \vdots\end{array} \right]\end{equation}
with $l=2N^2$ and $l\ll p=6N^2$ i.e. it is an over-determined system
and
\begin{itemize}
  \item  $\textbf{b}=[I_{L_{\rho_1,\phi_1}}, I_{L_{\rho_2,\phi_2}}, \dots,
I_{L_{\rho_p,\phi_p}}]^T$ are the observed measurements without any
additional noise.
  \item $\textbf{x}=
  \left[e_{x_{LS}}[1],\dots,e_{x_{LS}}[l/2],e_{y_{LS}}[1],\dots,
e_{y_{LS}}[l]\right]^T=[\textbf{E}_x|\textbf{E}_y]^T$ are the field
values to be recovered.
  \item $\bar{\textbf{A}}$ has the coefficients
\begin{equation*} \bar{a}_{ku} = \left\{ \begin{array}{lll}
         (\|\Delta
        L_{u}\|_k-\varepsilon_{ku}\Delta r)\cos\phi_k& \mbox{for $1\leq u\leq l/2$ if $\exists\,\,\Delta L_{uk}\subseteq L_{\rho_k,\phi_k}$}\\
        (\|\Delta L_{u-l/2}\|_k-\varepsilon_{ku-l/2}\Delta
        r)\sin\phi_k  & \mbox{for $l/2+1\leq u\leq l$ if $\exists\,\,\Delta L_{uk}\subseteq L_{\rho_k,\phi_k}$ } \\
        0 & \mbox{if $\not\exists\,\,\Delta L_{uk}\subseteq L_{\rho_k,\phi_k}$}    \end{array} \right.
        \end{equation*}
 and  $1\leq k \leq p$.
\end{itemize}
\paragraph{$\bar{\textbf{A}}$: the perturbed transfer matrix\\}
Matrix $\bar{\textbf{A}}$ of the linear system can be written as
$\bar{\textbf{A}}=\textbf{A}-\delta\textbf{A}$ where $\textbf{A}$ is
the ``unperturbed'' transfer matrix with elements
$a_{ku}=\bar{a}_{ku}-\varepsilon_{ku}\Delta r \cos\phi_k=\|\Delta
        L_{u}\|_k\Delta r \cos\phi_k$ for
$1\leq u\leq 1/2$ and $a_{ku}=\bar{a}_{ku}-\varepsilon_{ku}\Delta r
\sin\phi_k=\|\Delta L_{u}\|_k\Delta r\sin\phi_k$ for $l/2+1\leq
u\leq l$ and $\delta{\textbf{A}}$ is a perturbation of matrix
$\textbf{A}$ due to $\Delta r>0$ (eq.\ref{eq:samplingErrorFormula})
and has elements of the form $\varepsilon_{ku}\Delta r \cos\phi_k$
and $\varepsilon_{ku}\Delta r \cos\phi_k$.

The validation of the LS solution is performed by theoretically
estimating the relative error (RE)
\begin{equation}\label{relativeError}
\frac{\|\textbf{x}_{exact}-\bar{\textbf{x}}_{LS}\|}{\|\textbf{x}_{exact}\|}
\end{equation}
 using the $\|.\|$
\emph{Euclidean norm} \cite{MatrixAndLinearAlgebra} where
$\textbf{x}_{exact}$ is the column vector with the real values of
the field  while $\bar{\textbf{x}}_{LS}$ is the least square
solution of system \ref{eq:LSLinearSystem}.

The column vector $\textbf{x}_{exact}$ derives from the set of line
integrals
$I_{L_{\rho_1,\phi_1}},I_{L_{\rho_2,\phi_2}}\dots,I_{L_{\rho_p,\phi_p}}$
given by equation \ref{eq:lineIntegralWithErrorsFormulationsFinal}
which form the system
\begin{equation}
\textbf{b}=\textbf{A} \textbf{x}_{exact}+\delta\tilde{\textbf{E}}\
\label{system:UnperturbedCompactForm}\end{equation} where
$\delta\tilde{\textbf{E}}= [\delta\tilde{E}_{1},\dots
,\delta\tilde{E}_p]^T$ is a column vector with the discretization
error (eq.\ref{eq:discretizationError}) of equations
\ref{eq:lineIntegralWithErrorsFormulationsFinal}, $\textbf{b}$ are
the observed measurements and $\textbf{A}$ the ``unperturbed''
transfer matrix.

Next the following \emph{lemma} is proven.
\begin{lem}\label{RE_formulation}The relative solution error (RE) of systems
$\bar{\textbf{A}}\textbf{x}_{LS}=\textbf{b}$ and
$\textbf{A}\textbf{x}_{exact}+\delta\tilde{\textbf{E}}=\textbf{b}$
is given by
\begin{equation}\frac{\|\textbf{x}_{exact}-\bar{\textbf{x}}_{LS}\|}{\|\textbf{x}_{exact}\|}\leq
e_A\|\bar{\textbf{A}}^\dag\|_2\|\textbf{A}\|_2+e_Ak(\textbf{A})+\frac{e_b\|\bar{\textbf{A}}^\dag\|_2
\|\textbf{b}\|}{\|\textbf{x}_{exact}\|}
\end{equation}
where $\textbf{A}$ and $\bar{\textbf{A}}$ $\in \Re^{p\times l}$ with
$p>l$, $\bar{\textbf{A}}=\textbf{A}-\delta\textbf{A}$,
$\|\delta\textbf{A}\|_2=e_A\|\textbf{A}\|_2\neq\textbf{0}$ and
$\|\delta\tilde{\textbf{E}}\|=e_b\|\textbf{b}\|\neq\textbf{0}$ with
$Rank(\bar{\textbf{A}})\geq Rank(\textbf{A})=v\leq p$ (see Appendix
\ref{Appen:RankOfPerturbedMatrix}).

Moreover,
$\|\textbf{A}\|_2=max_{\|\textbf{x}\|=1}\|\textbf{A}\textbf{x}\|=\sqrt{\sigma_{max}}$
is the matrix \emph{2-norm} of $\textbf{A}$ with $\sigma_{max}$ the
maximum singular value of $\textbf{A}$ and
$\|\textbf{A}^\dag\|_2=\frac{1}{\sqrt{\sigma_{min}}}$ where
$\sigma_{min}$ the minimum singular value and $k(\textbf{A})$ the
condition number defined as
$k(\textbf{A})=\|\textbf{A}\|_2\|\textbf{A}^\dag\|_2$ and
 where symbol $\dag$
refers to the pseudo-inverse of the rectangular matrices
$\textbf{A}$ and $\bar{\textbf{A}}$ such as
$\textbf{A}^\dag=(\textbf{A}^T\textbf{A})^{-1}\textbf{A}^T$
\end{lem}
\begin{proof}
For the estimation of the relative error (RE) we employ the
decomposition theorem $\textbf{A}^\dag-\bar{\textbf{A}}^\dag$
discussed in \cite{Wedin} and \cite{Wei1989}.

According to the decomposition theorem \cite{Wedin},
\[\textbf{A}^\dag-\bar{\textbf{A}}^\dag=-\bar{\textbf{A}}^\dag(\textbf{A}-\bar{\textbf{A}})\textbf{A}^\dag-
\bar{\textbf{A}}^\dag(\textbf{I}-\textbf{A}\textbf{A}^\dag)+(\textbf{I}-\bar{\textbf{A}}^\dag\bar{\textbf{A}})\textbf{A}^\dag
\]
Therefore, if
$\bar{\textbf{b}}=\textbf{b}-\delta\tilde{\textbf{E}}$,\begin{equation}\label{eq:DecompositionTheorem_x_ls_x_exact}\begin{split}
 \textbf{x}_{exact}-\bar{\textbf{x}}_{LS}&=
\textbf{A}^\dag \bar{\textbf{b}}- \bar{\textbf{A}}^\dag
\textbf{b}=\textbf{A}^\dag \bar{\textbf{b}}- \bar{\textbf{A}}^\dag
(\bar{\textbf{b}}+\delta\tilde{\textbf{E}})=(\textbf{A}^\dag -
\bar{\textbf{A}}^\dag)\bar{\textbf{b}}-\bar{\textbf{A}}^\dag\textbf{e}
=\\&=
\left[-\bar{\textbf{A}}^\dag(\textbf{A}-\bar{\textbf{A}})\textbf{A}^\dag-
\bar{\textbf{A}}^\dag(\textbf{I}-\textbf{A}\textbf{A}^\dag)+
(\textbf{I}-\bar{\textbf{A}}^\dag\bar{\textbf{A}})\textbf{A}^\dag
\right]\bar{\textbf{b}}-\bar{\textbf{A}}^\dag\delta\tilde{\textbf{E}}
\end{split}
\end{equation}
$(\textbf{I}-\textbf{A}\textbf{A}^\dag)\bar{\textbf{b}}=\textbf{0}$
as $\bar{\textbf{b}}\in Range(\textbf{A})$ (i.e
$\bar{\textbf{b}}=\textbf{A}\textbf{x}_{exact}$).

So, equation \ref{eq:DecompositionTheorem_x_ls_x_exact} becomes
\begin{equation}\label{eq:DecompositionTheorem_x_ls_x_exact2}
\begin{split}
\textbf{x}_{exact}-\bar{\textbf{x}}_{LS}&=
\left[-\bar{\textbf{A}}^\dag(\textbf{A}-\bar{\textbf{A}})
\textbf{A}^\dag+(\textbf{I}-\bar{\textbf{A}}^\dag\bar{\textbf{A}})
\textbf{A}^\dag\right]\bar{\textbf{b}}-\bar{\textbf{A}}^\dag
\delta\tilde{\textbf{E}}=\\&=-\bar{\textbf{A}}^\dag
\delta\textbf{A}\textbf{x}_{exact}-\bar{\textbf{A}}^\dag
\delta\tilde{\textbf{E}}+(\textbf{I}-\bar{\textbf{A}}^\dag\bar{\textbf{A}})\textbf{A}^\dag
\bar{\textbf{b}}
\end{split}
\end{equation}

Moore Penrose pseudo-inverse identities
$\textbf{A}^\dag=\textbf{A}^\dag\textbf{A}\textbf{A}^\dag=
({\textbf{A}^\dag}\textbf{A})^T\textbf{A}^\dag=\textbf{A}^T(\textbf{A}^{\dag})^T\textbf{A}^\dag$
\cite{MatrixAndLinearAlgebra}\cite{Wedin} imply that
$\textbf{A}^\dag\bar{\textbf{b}}=\textbf{A}^T(\textbf{A}^{\dag})
^T\textbf{A}^\dag\bar{\textbf{b}}=\textbf{A}^T(\textbf{A}^{\dag})^T\textbf{x}_{exact}$.

Hence, equation \ref{eq:DecompositionTheorem_x_ls_x_exact2} becomes
\begin{equation}
\begin{split}\label{eq:DecompositionTheorem_x_ls_x_exact3}
\textbf{x}_{exact}-\bar{\textbf{x}}_{LS}&=-\bar{\textbf{A}}^\dag
\delta\textbf{A}\textbf{x}_{exact}-\bar{\textbf{A}}^\dag
\delta\tilde{\textbf{E}}+(\textbf{I}-\bar{\textbf{A}}^
\dag\bar{\textbf{A}})\textbf{A}^T(\textbf{A}^{\dag})^{T}\textbf{x}_{exact}=\\
&=-\bar{\textbf{A}}^\dag
\delta\textbf{A}\textbf{x}_{exact}-\bar{\textbf{A}}^\dag
\textbf{e}+(\textbf{I}-\bar{\textbf{A}}
^\dag\bar{\textbf{A}})(\bar{\textbf{A}}+\delta\textbf{A})^T(\textbf{A}^{\dag})^{T}\textbf{x}_{exact}\\
\end{split}
\end{equation}
$(\textbf{I}-\bar{\textbf{A}}^\dag\bar{\textbf{A}})\bar{\textbf{A}}^T=\textbf{0}$
i.e. $\bar{\textbf{A}}^T \in \,\, N
(\bar{\textbf{A}}^\dag\bar{\textbf{A}})$
as
$\bar{\textbf{A}}^T=\bar{\textbf{A}}^\dag\bar{\textbf{A}}\bar{\textbf{A}}^T$\cite{MatrixAndLinearAlgebra}.

Thus,
\begin{equation}
\begin{split}\label{eq:DecompositionTheorem_x_ls_x_exact4}
 \textbf{x}_{exact}-\bar{\textbf{x}}_{LS}&=-\bar{\textbf{A}}^\dag
\delta\textbf{A}\textbf{x}_{exact}-\bar{\textbf{A}}^\dag
\delta\tilde{\textbf{E}}+(\textbf{I}-\bar{\textbf{A}}^\dag\bar
{\textbf{A}})\delta\textbf{A}^T(\textbf{A}^{\dag})^{T}\textbf{x}_{exact}\end{split}
\end{equation}
According to the matrix norm inequalities of \emph{2-norm} $\|.\|_2$
\cite{MatrixAndLinearAlgebra} for matrices $\textbf{B}$ and
$\textbf{C}$ $\in\,\,\Re^{m\times n}$
$\|\textbf{B}\textbf{C}\|_2\leq\|\textbf{B}\|_2\|\textbf{C}\|_2$ and
$\|\textbf{B}+\textbf{C}\|_2\leq\|\textbf{B}\|_2+\|\textbf{C}\|_2$
and $\|\textbf{B}\textbf{x}\|_2\leq\|\textbf{B}\|_2\|\textbf{x}\|$
when $\textbf{x}$ is a vector $\in\,\,\Re^{n\times 1}$.

So,
\begin{equation*}\begin{split}
\frac{\|\textbf{x}_{exact}-\bar{\textbf{x}}_{LS}\|}{\|\textbf{x}_{exact}\|}&=
\frac{\|-\bar{\textbf{A}}^\dag \delta\textbf{A}\textbf{x}_{exact}
+(\textbf{I}-\bar{\textbf{A}}^\dag\bar{\textbf{A}})\delta\textbf{A}^T
(\textbf{A}^{\dag})^{T}\textbf{x}_{exact}-\bar{\textbf{A}}^\dag
\delta\tilde{\textbf{E}}\|_2}{\|\textbf{x}_{exact}\|}\leq\\&\leq
\frac{\|-\bar{\textbf{A}}^\dag
\delta\textbf{A}\textbf{x}_{exact}\|_2+\|(\textbf{I}-
\bar{\textbf{A}}^\dag\bar{\textbf{A}})\delta\textbf{A}^T(\textbf{A}^{\dag})^{T}\textbf{x}_{exact}
\|_2+\|-\bar{\textbf{A}}^\dag
\delta\tilde{\textbf{E}}\|_2}{\|\textbf{x}_{exact}\|}\leq\\&\leq
\|\delta\textbf{A}\|_2 \|\bar{\textbf{A}}^\dag\|_2+
\|(\textbf{I}-\bar{\textbf{A}}^\dag\bar{\textbf{A}})\|_2\|\delta\textbf{A}^T\|_2\|\textbf{A}^{\dag}\|_2
+\frac{\|\bar{\textbf{A}}^\dag\|_2
\|\delta\tilde{\textbf{E}}\|}{\|\textbf{x}_{exact}\|}\leq\\&\leq\|\delta\textbf{A}\|_2
\|\bar{\textbf{A}}^\dag\|_2+\|\delta\textbf{A}\|_2\|\textbf{A}^{\dag}\|_2+\frac{\|\bar{\textbf{A}}^\dag\|_2
\|\delta\tilde{\textbf{E}}\|}{\|\textbf{x}_{exact}\|}
 \end{split}
\end{equation*}
if we set $\|\delta\textbf{A}\|_2=e_A\|\textbf{A}\|_2$,
$\|\delta\tilde{\textbf{E}}\|=e_b\|\textbf{b}\|$ and
$k(\textbf{A})=\|\textbf{A}^\dag\|_2\|\textbf{A}\|_2$ the condition
number of $\textbf{A}$ then \[ \boxed{
\frac{\|\textbf{x}_{exact}-\bar{\textbf{x}}_{LS}\|}{\|\textbf{x}_{exact}\|}\leq
e_A\|\bar{\textbf{A}}^\dag\|_2\|\textbf{A}\|_2+e_Ak(\textbf{A})+\frac{e_b\|\bar{\textbf{A}}^\dag\|_2
\|\textbf{b}\|}{\|\textbf{x}_{exact}\|} }\]
\end{proof}

\subsubsection{Conditions for Bounded Relative Error}
The relative error (RE) is bounded when the sampling error is small
i.e. $e_A\rightarrow0$ and the discretization error
$\|\delta\tilde{\textbf{E}}\|=e_b\|\textbf{b}\|>\textbf{0}$ with
$0<e_b\ll1$.

If the sampling step $\Delta r\rightarrow 0$ then
$\delta\textbf{A}\approx\textbf{0}$ and
$\bar{\textbf{A}}\approx\textbf{A}$ (the sampling error is
eliminated $\varepsilon(\Delta r)\simeq 0$) and the relative error
RE is
\[\frac{\|\textbf{x}_{exact}-\bar{\textbf{x}}_{LS}\|}{\|\textbf{x}_{exact}\|}\leq
\frac{e_b\|{\textbf{A}}^\dag\|_2
\|\textbf{b}\|}{\|\textbf{x}_{exact}\|}\leq
\frac{e_b}{1-e_b}\|{\textbf{A}}^\dag\|_2\|{\textbf{A}}\|_2=\frac{e_b}{1-e_b}k(\textbf{A})\]
%
When the discretization error $\|\delta\tilde{\textbf{E}}\|_2>0$
with $\|\delta\tilde{\textbf{E}}\|=e_b\|\textbf{b}\|$ and
$0<e_b\ll1$ the equations of the LS system \ref{eq:LSLinearSystem}
have the form (eq.\ref{eq:ApproximationOflineIntegral})
\[\sum_{u=1}^l\{\|\Delta L_{u}\|
(e_{x_{LS}}[u]\cos\phi + e_{y_{LS}}[u]\sin\phi)\}\] where $\|\Delta
L_u\|>0$. This formulation is an approximation of the line integral
equation (eq.\ref{eq:lineIntegralWithErrorsFormulationsFinal}) and
the summation of any set of these equations cannot be zero in any
closed path. So, the rectangular transfer matrix of LS system
\ref{eq:LSLinearSystem} is not rank deficient as there are no
linearly dependent equations and the condition number
$k(\textbf{A})<\infty$.

Thus, the relative error (RE)
\[\frac{\|\textbf{x}_{exact}-\bar{\textbf{x}}_{LS}\|}{\|\textbf{x}_{exact}\|}\leq\frac{e_b}{1-e_b}k(\textbf{A})<\infty\]
is bounded.

If the resolution of the discrete domain improves (grid refinement)
such as the discretization error
$\|\delta\tilde{\textbf{E}}\|_2\simeq0$, the linear equations are
\[I_{L_{\rho,\phi}}=\lim_{\|\Delta L_u\|\rightarrow 0}\sum_{u}\left\|\Delta
L_{u}\|(e_x[u]\cos\phi+e_y[u]\sin\phi)\right\}\rightarrow\int_{L_{\rho,\phi}}e_x\cos\phi
dl+e_y\sin\phi dl\] then the system becomes severely ill
conditioned. The condition number $k(\textbf{A})\rightarrow\infty$
due to the linear dependencies between the equations and the RE
bound tends to infinity.

Particulary, for a grid $N\times N$ ($N$ is the resolution), $4N$
boundary values and $2N^2$ unknowns, when $N\rightarrow\infty$ (i.e.
grid refinement), the number of independent equations
$F(N)\rightarrow 4(N-1)$ as the equations approximate the continues
line integrals (section \ref{subsection:IllPosednessConsideration}).
Then, the number of independent equations cannot exceed the number
of the unknowns as
$\lim_{N\rightarrow\infty}\frac{F(N)}{2N^2}=\lim_{N\rightarrow\infty}\frac{4(N-1)}{2N^2}=0$
and  system \ref{eq:LSLinearSystem} becomes rank deficient and
consistent (i.e. infinity number of solutions).

In other words rank deficiency of transfer matrix $\textbf{A}$ (e.g
$Rank(\textbf{A})=r$) implies that there are singular
values$\sigma_r\rightarrow 0$. So according to singular value
decomposition (SVD) $\textbf{A}=\textbf{U}\Sigma\textbf{V}^\dag$
where matrix $\textbf{U}$ and matrix $\textbf{V}$ are unitary, and
$\Sigma$ is an matrix whose only non-zero elements are along the
diagonal with $\sigma_i>\sigma_{i+1}\geq 0$(The columns of
$\textbf{U}$ and  $\textbf{V}$ are known as the left
$\{\textbf{u}_i\}$ and right $\{\textbf{v}_i\}$ singular vectors)
for $\sigma_r=0\,\,\Rightarrow$ $\textbf{A}u_r=0$ and thus
$\textbf{A}(\textbf{x}+a\textbf{u}_r)=\textbf{b}$ which show that
the uniqueness test fails($2^{nd}$ Hadamards criterion).

So, a theoretical relationship between the discretization and the
ill-posedness of the inverse problem has been presented.
Particularly, the discretization process is a way to regularize the
continuous ill posed problem since the discretization ensures a
finite upper bound to the solution error. This is called
self-regularization (regularization by projection) or regularization
by discretization. The proper choice of the discretization
parameters is important for the problem regularization. An extreme
coarse discretization of the domain increases the ill conditioning
of the problem in the sense that if $e_b\rightarrow1$ (high
discretization error), the RE is again unbounded, leading again to
an unstable linear system. Graphically this is presented in figure
\ref{fig:REvsDiscretizationCoefficient}
\begin{figure}[!htb]
\centering
\includegraphics[width=0.3\textwidth]{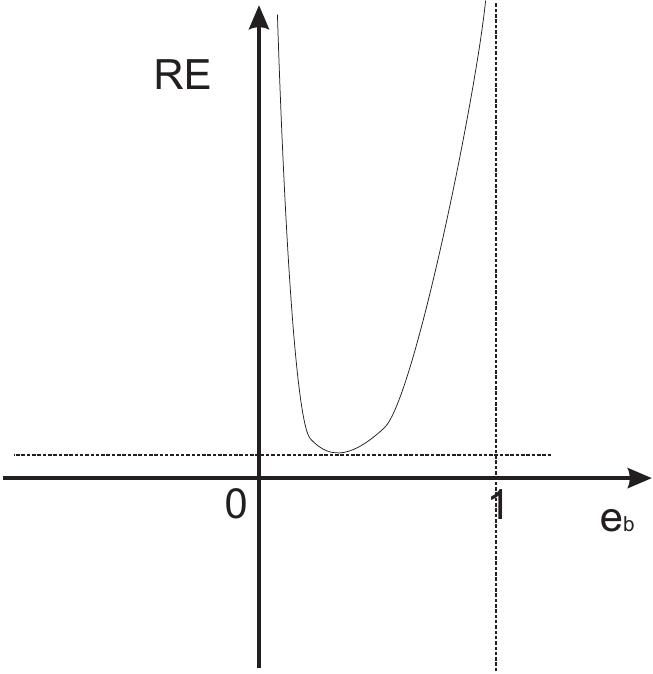} \caption{Relative Error RE vs
discretization coefficient
$e_b$}\label{fig:REvsDiscretizationCoefficient}
\end{figure}
So, a bounded discretization error $\delta\tilde{\textbf{E}}$ is
essential for an accurate solution to balance the solution
approximation error and the ill conditioning of the linear system.

\section{Summary}
The proposed method intends to approximate a vector field from a set
of line integrals. As we showed, from $4N$ boundary measurements,
the maximum number of equations which can be obtained taking all
possible pairs of boundary measurements is  $2N(4N-1)$ and the
maximum number of independent line integrals cannot exceed the
$4N-1$ equations. In the case of a square domain with $2N^2$
unknowns, this number of independent continuous line integrals is
not sufficient for the recovery of a $2D$ irrotational vector field.
However, we provided a theoretical proof showing that the
discretization can be an efficient way of regularizing the
continuous ill posed problem and therefore the problem is tractable
by solving it in the discrete domain.

\chapter{Simulations and a Real Application}
The current chapter is divided into two main parts: in the first
part the validation and verification of the line integral method is
presented performing some basic simulations while in the second part
an introduction to the inverse bioelectric field problem is reported
as a future real application of the proposed vector field
reconstruction method.

In particular, for the evaluation of the method we perform
simulations for the approximately reconstruction of an electrostatic
field produced by electric monopoles in a $2D$ square domain based
on the vector field recovery method. According to the preconditions
and assumptions of the mathematical model, the electrostatic field
is a good example for the validation of the method as it satisfies
the quasi static condition and the irrotational property.

Moreover, the numerical implementation of the method is based on the
modeling described in subsection
\ref{subsection:NumericalImplementation} where nearest neighbor
approximation was employed. For the verification of the theoretical
predictions, the magnitude of the singular values of transfer matrix
$\textbf{A}$ of system \ref{eq:SystemOfLinearEquations1}, will
provide a measure of the invertibility and conditioning of the
numerical system \ref{eq:SystemOfLinearEquations1}. The sampling
step along the line integrals will be very small (much smaller than
the cell size) in all simulations.

Finally, the proposed vector field method will be presented as an
equivalent mathematical counterpart of the partial differential
formulation of the inverse bioelectric field problem
\cite{BFP_johnson1993}. Comparisons of the two approaches will be
made.

\section{Results}
The qualitative and quantitative evaluation of the method is very
important for examination of its robustness. The validation is
concerned with how the mathematical and geometric formulations
represent the real physical problem and the verification assesses
the accuracy with which the numerical model approximates the real
one. The mathematical and geometric properties of the experimental
model are defined in the ``Simulation Setup'' subsection while in
subsection ``Simulations using Electric Monopoles'' the evaluation
of the previous theoretical findings is performed.

\subsection{Simulation Setup}
\paragraph {\emph{Geometrical Model}\\}
In the current simulations we intend to recover a vector field in a
$2D$ square bounded domain $\Omega$. For the numerical estimation of
the field, the discretization of $\Omega$ is essential.
Discretization of the domain can be described as
\[\Omega:\{(x,y)\in[-U,U]^2\}\underbrace{\rightarrow}_{Discretization}\{[i,j]\in [1:N,1:N]\}\]
where $i=\left\lfloor \frac{x+U}{P} \right\rfloor$, $j=\left\lfloor
\frac{y+U}{P}\right\rfloor$ (nearest neighbor approximation) with
$[P\times P]$ the cell's size and $N=2U/P$ the spatial resolution
(sampling rate) of the discrete domain.
\paragraph{\emph{Mathematical Model}\\}
The numerical representation of the line integral equations is given
by
\begin{equation}\label{eq:NumericalIntegrals}
\tilde{I}_{L_{\rho,\phi}}=\sum  N_{ij}
\textbf{E}[i,j]\cdot(\cos\phi,\sin\phi)\Delta r\end{equation} where
$N_{ij}$ is the number of samples in cell $[i,j]$ of the domain and
$\Delta r$ the sampling step. The linear system of equations is
designed according to subsection \ref{subsection:Geometric Model}.

\paragraph{\emph{Source Model}\\}
A set of monopoles is employed for the production of the real
electric field. The field created by monopoles (point sources) is
given by $\textbf{E}=k\sum_i\frac{Q_i}{(r-r_i)^2}\hat{\textbf{r}}_i$
where $(r-r_i)^2$ is the distance between charge $Q_i$ and the
E-field evaluation point $\textbf{r}$ and $\hat{\textbf{r}}_i$ the
unit vector pointing from the particle with charge $Q_i$ to the
E-field point. $\Phi=k\sum_i\frac{Q_i}{|r-r_i|}$ is the potential
function for the estimation of the potential values at the
boundaries of the domain.

So, in the inverse vector field problem, one seeks to estimate field
$\textbf{E}[i,j]$ in each cell $[i,j]$ when the potential values in
the middle of the boundary edges of the boundary cells are known.

The linear system \ref{eq:SystemOfLinearEquations1} is formulated
from a set of approximated line integrals
(\ref{eq:NumericalIntegrals}), where value
$\tilde{I}_{L_{\rho,\phi}}$ is the potential difference between two
boundary points.

In most simulations, the point sources of the field were positioned
outside the bounded domain in order to avoid singularities, since in
practical problems the value of the field cannot be infinite.

In the next subsection we examine the following scenarios:

\begin{itemize}
  \item increase the resolution of the interior of the discrete
  domain and the observed boundary measurements and examine the relationship between the resolution and the ill conditioning of the system \ref{eq:SystemOfLinearEquations1};
  \item create a field using more point sources with arbitrary charges and positions
  but still outside the bounded domain for constant resolution;
  \item place the point source inside the domain.
\end{itemize}

The goal is to examine experimentally the ill conditioning of the
system for different source distributions (close and far from the
bounded domain) and the relationship between the spatial resolution
of the domain and the ill conditioning. For the examination of the
ill conditioning, the singular values of transfer matrix
$\textbf{A}$ of the linear system \ref{eq:SystemOfLinearEquations1}
are estimated performing the
singular  value decomposition (SVD) of $\textbf{A}$. 
It is known that a slowly decreasing singular value spectrum with a
broader range of nonzero singular values indicates a better
conditioning while a rapidly decreasing to zero shows increase of
ill conditioning. Moreover, a second indicator is the condition
number $k(\textbf{A})=\frac{\sigma_{max}}{\sigma_{min}}$ where
$\sigma_{max}$ and $\sigma_{min}$ are the maximum and minimum
singular values of $\textbf{A}$, respectively. The condition number
is a gauge of the transfer error from matrix $\textbf{A}$ and vector
$\textbf{b}$ (eq. \ref{eq:SystemOfLinearEquations1}) to the solution
vector $\textbf{x}$. When the condition number is close to 1, the
conditioning of the system is good and the transfer error is low.
So, small changes in $\textbf{A}$ or $\textbf{b}$ produce small
errors in $\textbf{x}$. Finally, we estimate the relative error of
the magnitude and the phase between the approximated and the real
field using \begin{equation}\label{eq:RelativeError}
RE=\frac{\|\textbf{x}_{LS}-\textbf{x}_{exact}\|_2}{\|\textbf{x}_{exact}\|_2}\end{equation}
where $\|\textbf{x}\|_2=\sqrt{\sum_i x_i^2}$.

\subsection{Simulations using Electric Monopoles}
\begin{itemize}
  \item
In the first set of simulations the sources are selected to be far
from the recovery domain in order the field to be smooth and thus to
examine only the relationship between resolution and ill
conditioning due to grid refinement.

\paragraph{$1^{st} Example\\$}
For a bounded domain $\Omega:\{(x,y)\,\, \in \,\,[-U,U]^2\subset
\Re^2\}=[-5.5,5.5]^2$ and two point sources with the same charge $Q$
placed at $(-19,0)$ and $(19,19)$ on the $x,y$ coordinate system
(far from the domain $\Omega$) and small sampling step $\Delta
r=10^{-4}U$ along the integral line, we obtain the following.

When $P\times P=1\times 1$, the cell size and grid is $11\times11$
(resolution is 11), the recovered field is depicted in figure
\ref{fig:RecoveryFieldwithP=1}B.
\begin{figure}[!htb] \centering
\includegraphics[width=0.8\textwidth]{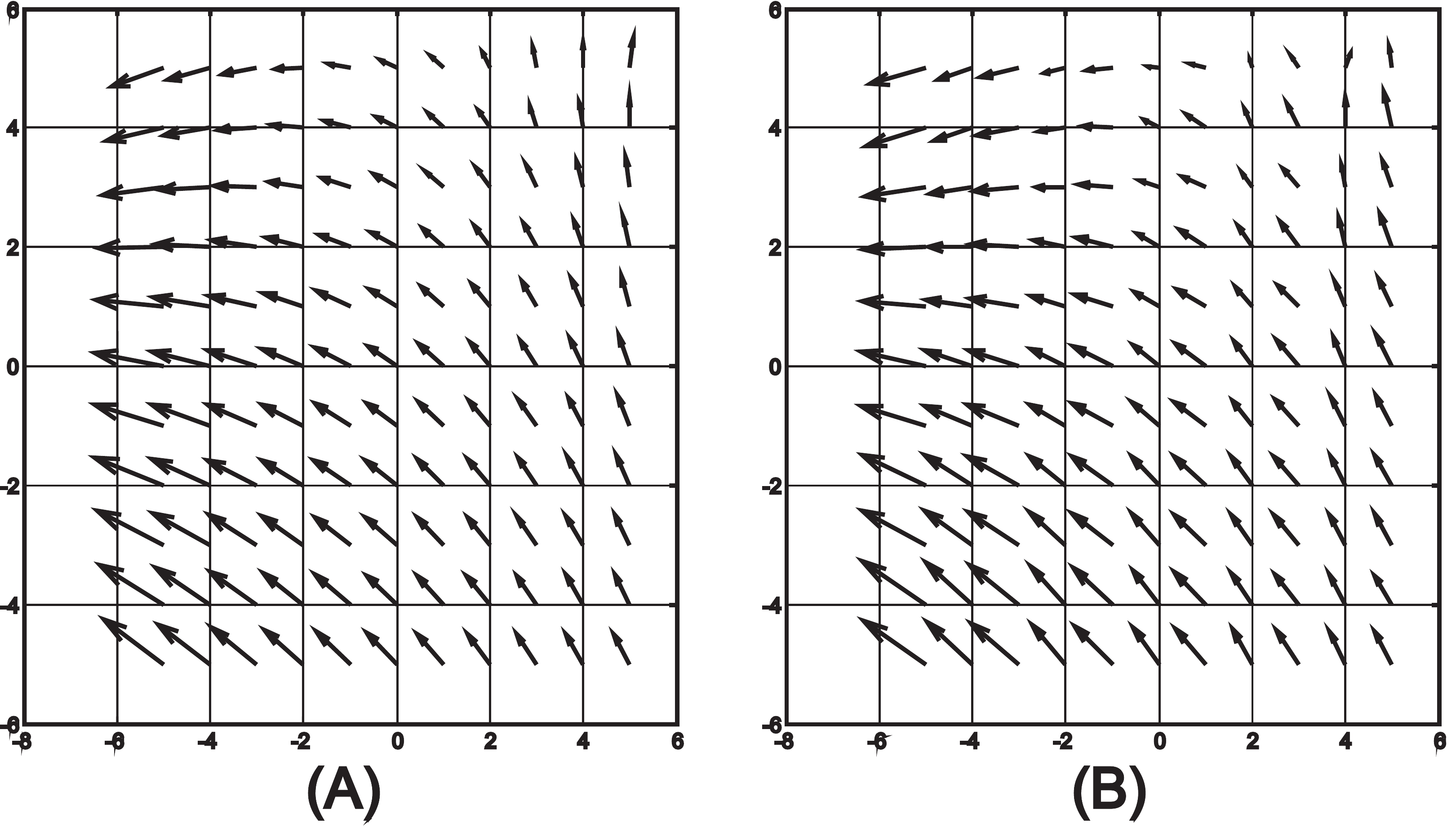} \caption{(A) Real Field and (B) Recovered Field when the grid resolution is $N=2U/P=11.$}\label{fig:RecoveryFieldwithP=1}
\end{figure}
The condition number of the transfer matrix $\textbf{A}$ of the
linear system is $k(\textbf{A})=174$ and the relative errors (eq.
\ref{eq:RelativeError}) of the magnitude and the phase between the
real and the reconstructed field are $RE_{magnitude}=0.11$ and
$RE_{phase}=0.05$ respectively.

\begin{figure}[!htb] \centering
\includegraphics[width=0.8\textwidth]{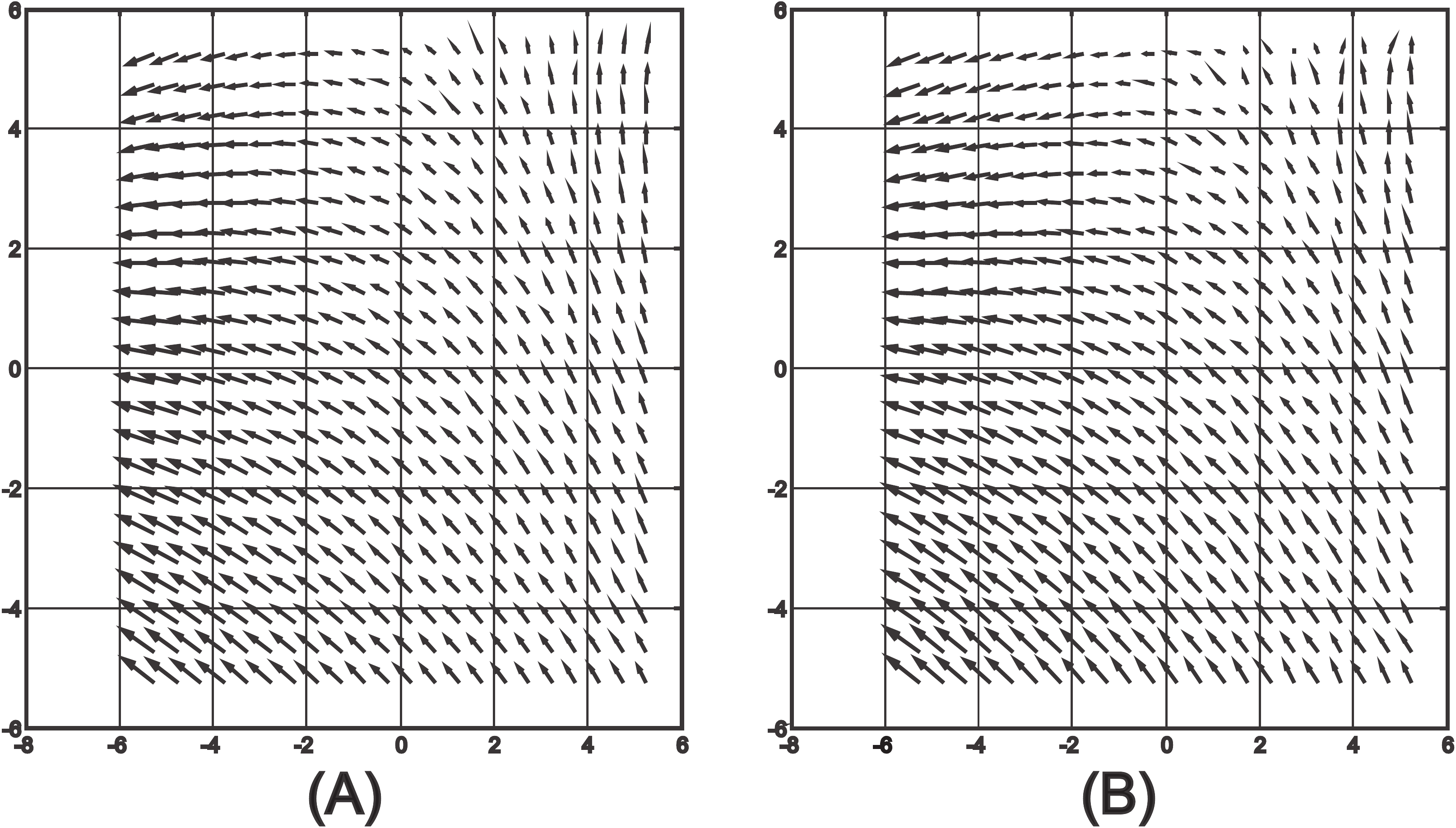} \caption{Refining the grid resolution, (A)
Real Field and (B) Recovered Field when the grid resolution is
$N=2U/P=22$.}\label{fig:RecoveryFieldwithP=0_5}
\end{figure}

When the cell size is $P=0.5$, the condition number is $
k(\textbf{A})=10^4$ and the relative errors $RE_{magnitude}=0.08$
and $RE_{phase}=0.045$. The recovered field is showed in figure
\ref{fig:RecoveryFieldwithP=0_5}B. In both cases $\Delta r$ is much
less than the cell size $P$. As we can observe the condition number
in the second case in much higher than in the first case where the
spatial resolution is ``coarse''. This implies that the solution of
 the second system $\textbf{A}\textbf{x}_{LS}=\textbf{b}$ is more sensitive and unstable to
$\textbf{A}$ or $\textbf{b}$ perturbations.

System instability means that the continuous dependence of the
solution upon the input data cannot be guaranteed and in the
presence of input noise (perturbations) the system behavior is
unpredictable. In the current simulations there is no additional
noise or other external sources of error, so the ill conditioning of
the system (intrinsic ill conditioning) due to the high condition
number $k(\textbf{A})$ of order $10^k$ has an effect on the computed
solution in the sense that a loss of $k$ digit accuracy in the
solution is applied (rule of thumb \cite{MatrixAndLinearAlgebra}).
Thus, for a floating point arithmetic (16 digits floating point
numbers are used generally in these simulations) only $16-k$ digit
solution accuracy can be achieved. So, in the absence of sources of
noise, there is mainly a loss in accuracy when the condition number
is high while there no great effect on the system's stability.

The approximation error RE is lower for the ``refined'' system (Res.
22) than that of the ``coarse'' system (Res. 11). In the ``refined''
system, there is a loss of $4$ digits in the solution accuracy which
is not so high for floating point measurements and due to grid
refinement, more spatial frequencies of the field can be recovered.
Thus, the solution of the ``refined'' system is more accurate.
However, the ``refined'' system is more unstable and if the spatial
resolution of the problem tends to the continuous case then
obviously the linear system will tend to rank deficiency.

\begin{figure}[!htb] \centering
\includegraphics[width=0.5\textwidth]{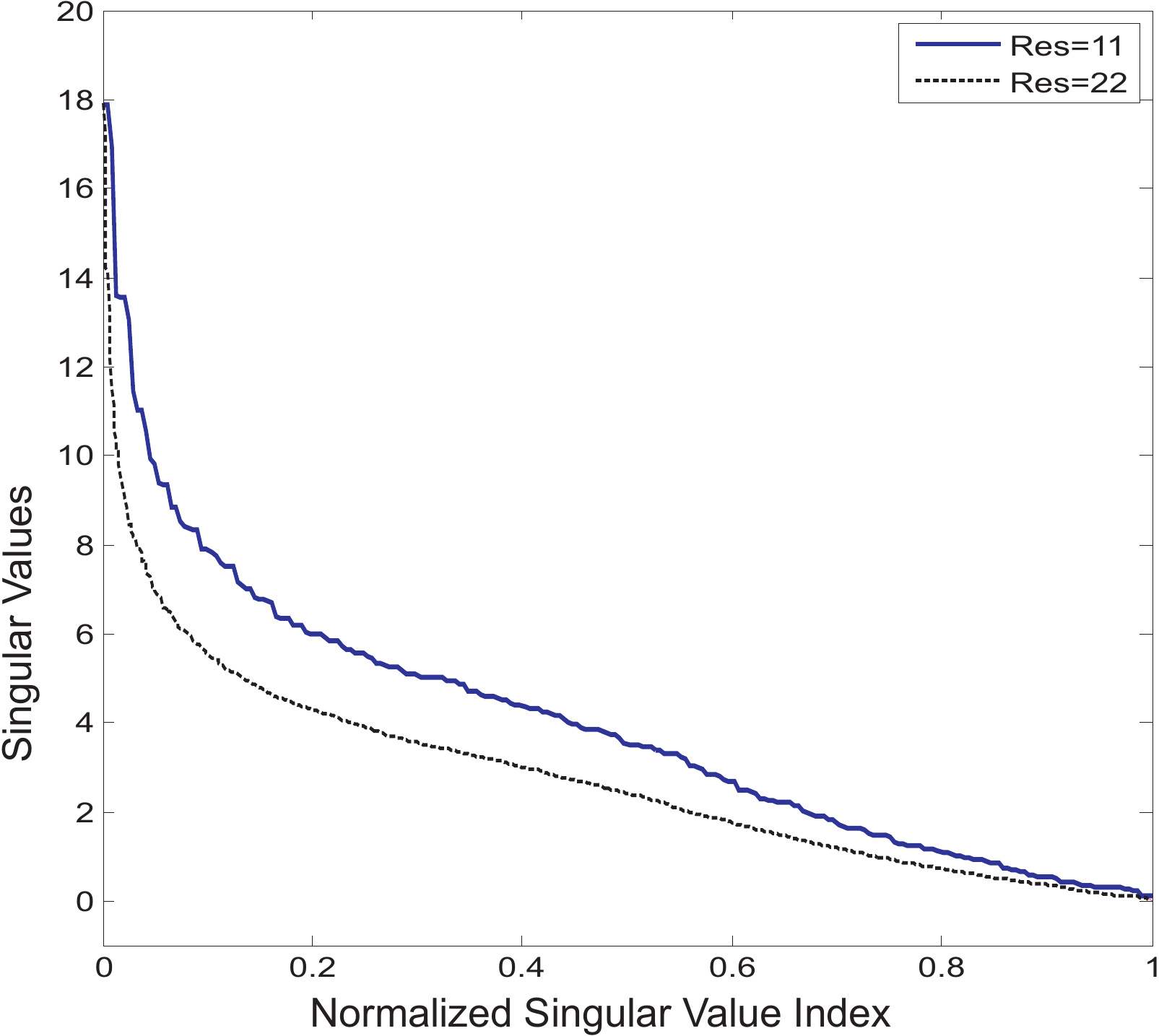} \caption{Refining the grid resolution, the spectrum of the singular
values (dotted line) descents more rapidly to zero.
}\label{fig:SingularValuesSpecturmOnly2}
\end{figure}

Moreover, grid refinement results the faster descent of the singular
values of the spectrum to zero and frequently a narrower range of
non zero singular values, which also indicate that the conditioning
of the transfer matrix $\textbf{A}$ and the stability of the linear
system deteriorates. This is clear according to figure
\ref{fig:SingularValuesSpecturmOnly2} where the continuous line
depicts the singular values spectrum when the grid resolution is
$N=U/P=11$ and the dashed line the spectrum for $N=22$.

\paragraph{$2^{nd} Example\\$}
Further results for $2$ point sources with $Q=10^{-8}$ localized
again at $(-19,0)$ and $(19,19)$ (far away from the bounded domain
$\Omega:\{(x,y)\in[-2.5,2.5]^2\}$) and 3 different cell sizes
$P=1,0.5,0.25$  are presented in table 4.1 and figure
\ref{fig:SingularValuesSpecturm}.
\begin{figure}[!htb] \centering
\includegraphics[width=0.8\textwidth]{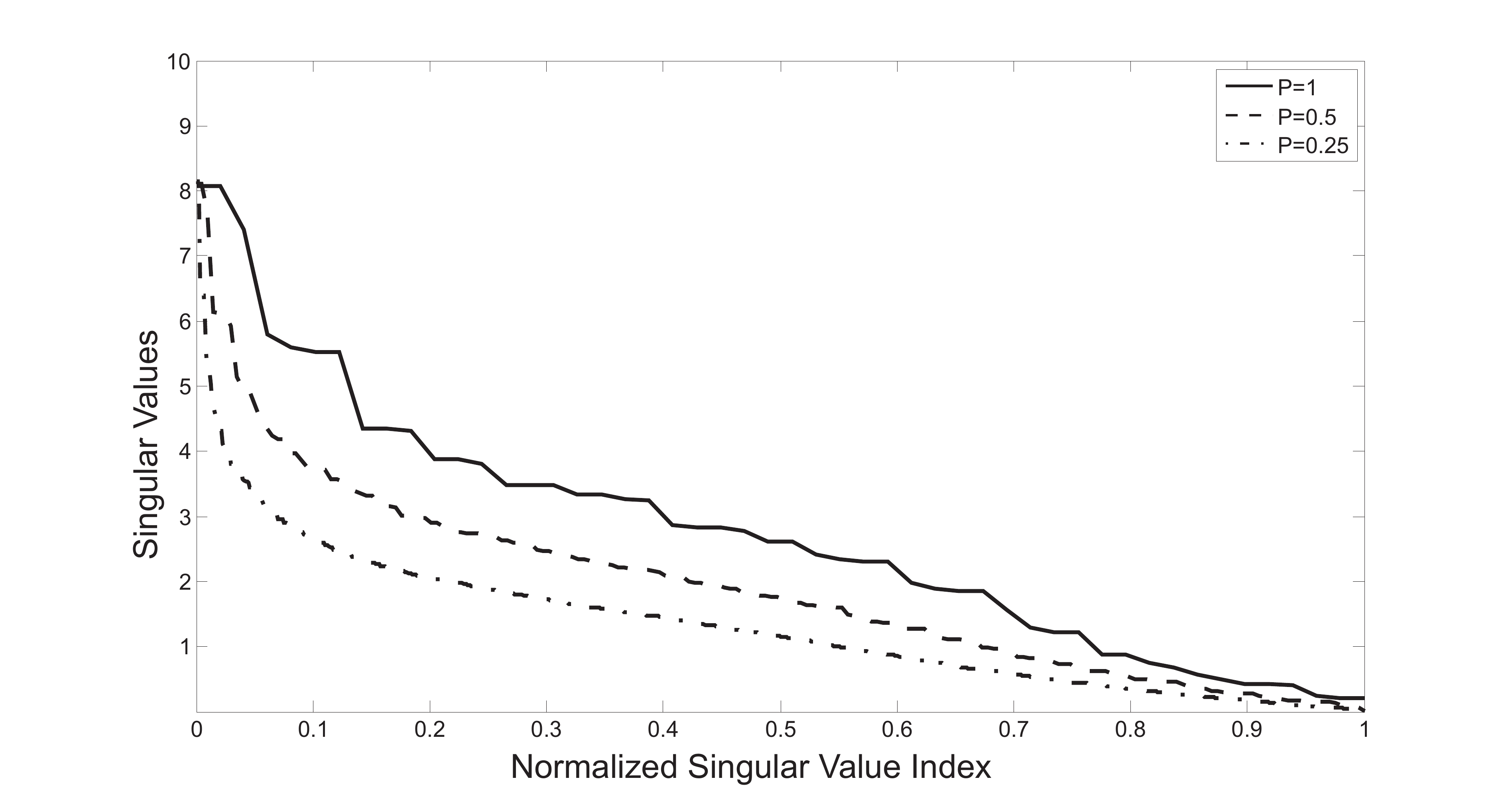} \caption{When refining the grid resolution, the spectrum of the singular
values descents more rapidly to zero and thus the conditioning of
the system deteriorates. In this example, 2 point sources with
$Q=10^{-8}$ localized at $(-19,0)$ and $(19,19)$ (far away from the
bounded domain $\Omega:\{(x,y)\in[-2.5,2.5]^2\}$) and 3 different
cell sizes were used $P=1,0.5,0.25$.
}\label{fig:SingularValuesSpecturm}
\end{figure}

\begin{table}[!htb]
\begin{center}
\begin{tabular}[!htb]{|c|c|c|c|c|}
  \hline
   Resolution  & Cell Size & \multicolumn{3}{|c|}{Measurements } \\ \hline
   $N\times N$ & P & k(\textbf{A}) & $RE_{Magnitude}$&$RE_{Phase}$  \\
   \hline\hline
   $20\times 20$& $0.25$ & $219.2 $& 0.0168& 0.0048  \\ \hline
   $10\times 10$& $0.5$  &$93.16 $ & 0.024 & 0.008 \\ \hline
   $5\times 5$  & $1$    &$96.7$        & 0.044 & 0.013 \\
   \hline
\end{tabular}\caption{Condition number $k(\textbf{A}$) and relative errors for different resolution levels}
\end{center}\label{table:1}
\end{table}

\newpage
\item
When a large number of point sources (160 sources with arbitrary
charges) are placed far from domain $\Omega$
(fig.\ref{fig:ManyPointSources}) then $RE_{magnitude} = 0.1567$
$RE_{phase} = 0.1527$ and the recovered field is depicted in figure
\ref{fig:ManySourcesRecovery}.
\begin{figure}[!htb] \centering
\includegraphics[width=0.7\textwidth]{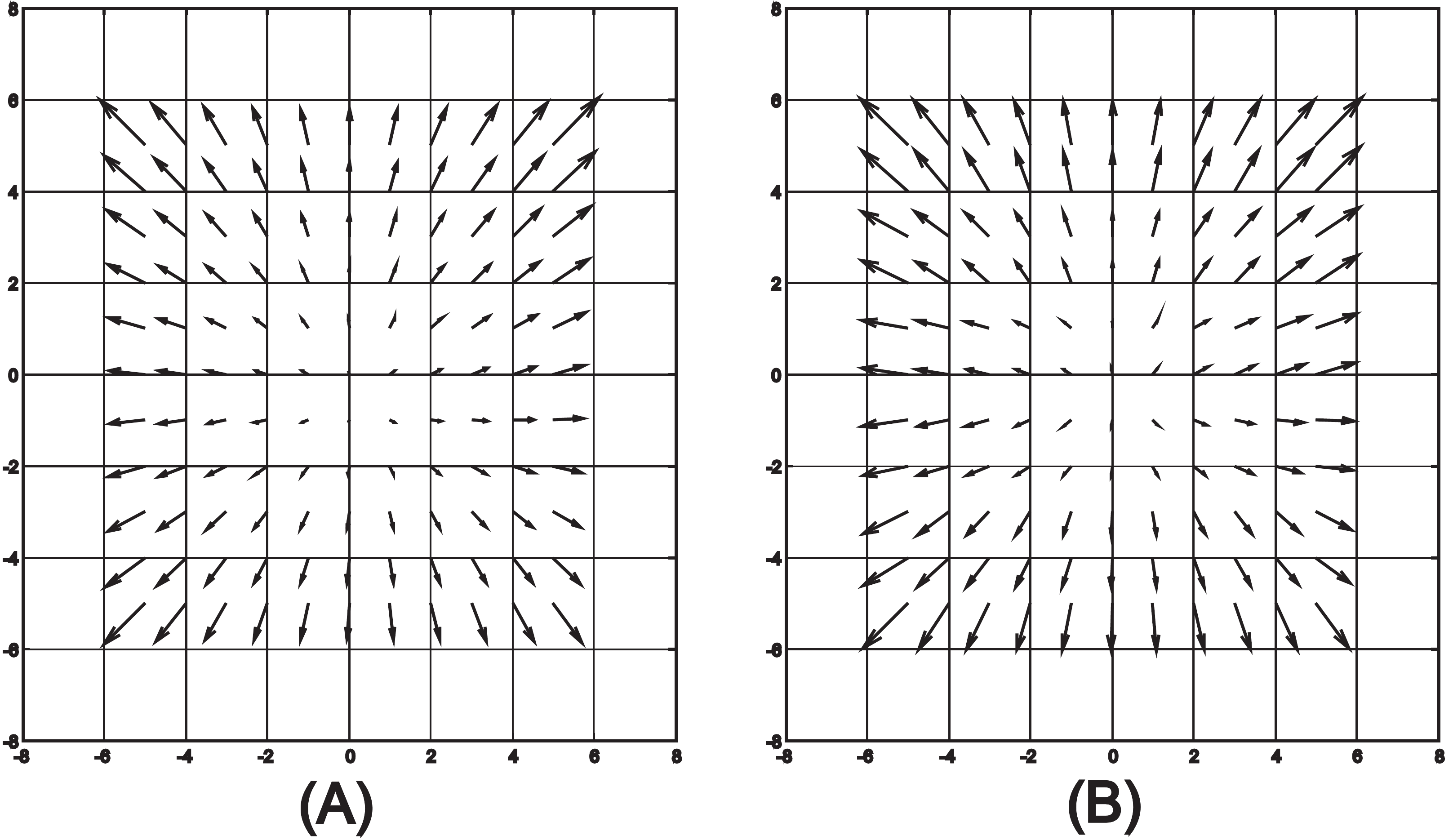} \caption{(A) Real Field and (B) Recovered Field
when the grid resolution is $N=2U/P=11$ and the real field is
produced from
 many point sources which are around the vector field domain.}\label{fig:ManySourcesRecovery}
\end{figure}
\begin{figure}[!htb] \centering
\includegraphics[width=0.4\textwidth]{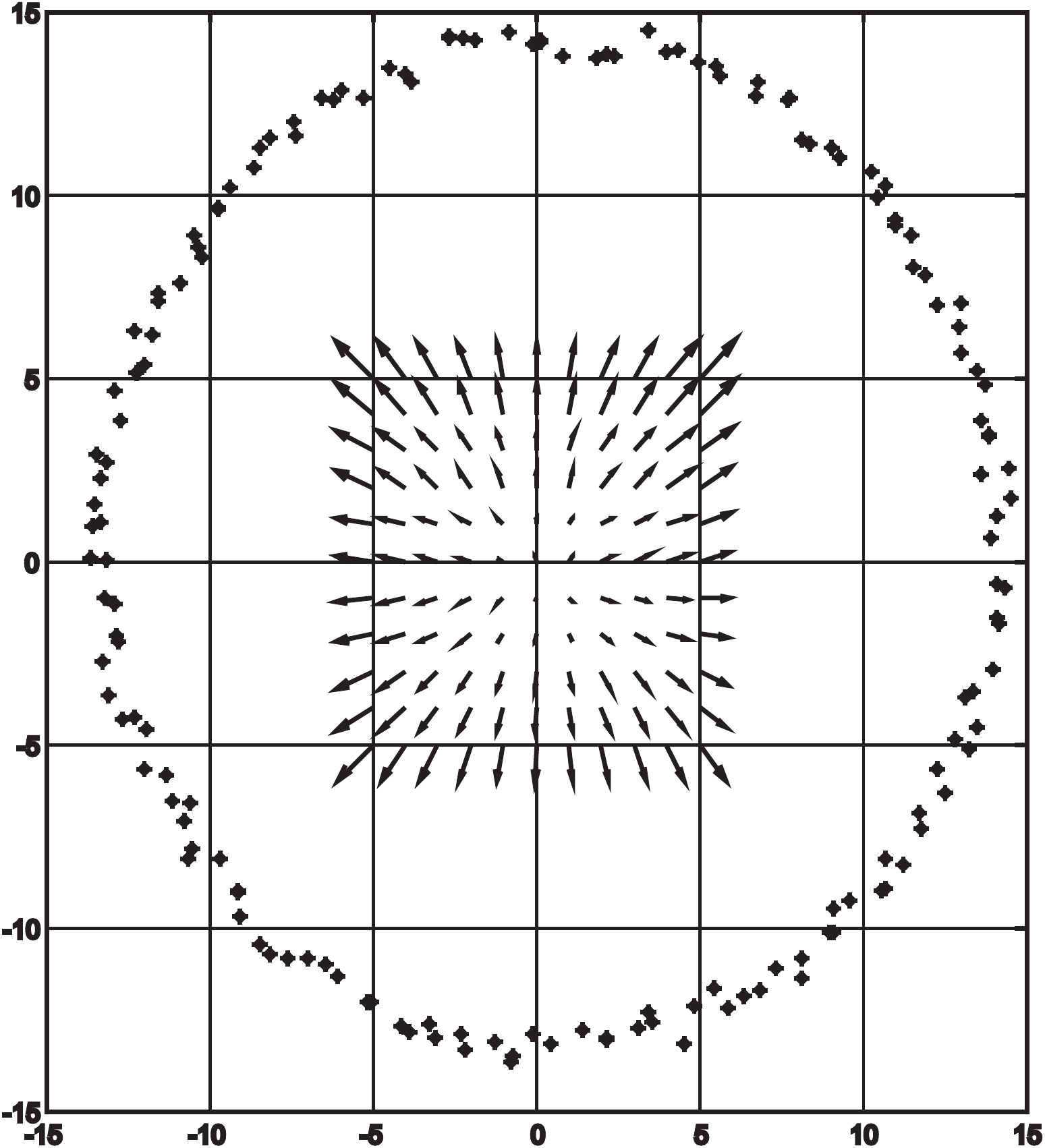} \caption{Dots depict the locations of the point
sources far from the bounded domain
 and the arrows the recovered vector field inside the bounded domain}\label{fig:ManyPointSources}
\end{figure}

If the point sources are closer to the field of interest, like in
figures \ref{fig:ManyPointSourcesCloseToDomain} and
\ref{fig:ManySourcesDistributionCloseToDomain}, and keeping constant
resolution, the relative errors are $RE_{magnitude}=0.63$ and
$RE_{phase}=0.33$. Clearly, the relative errors increase as the
field sources are placed closer to the recovery domain.
\begin{figure}[!htb] \centering
\includegraphics[width=0.7\textwidth]{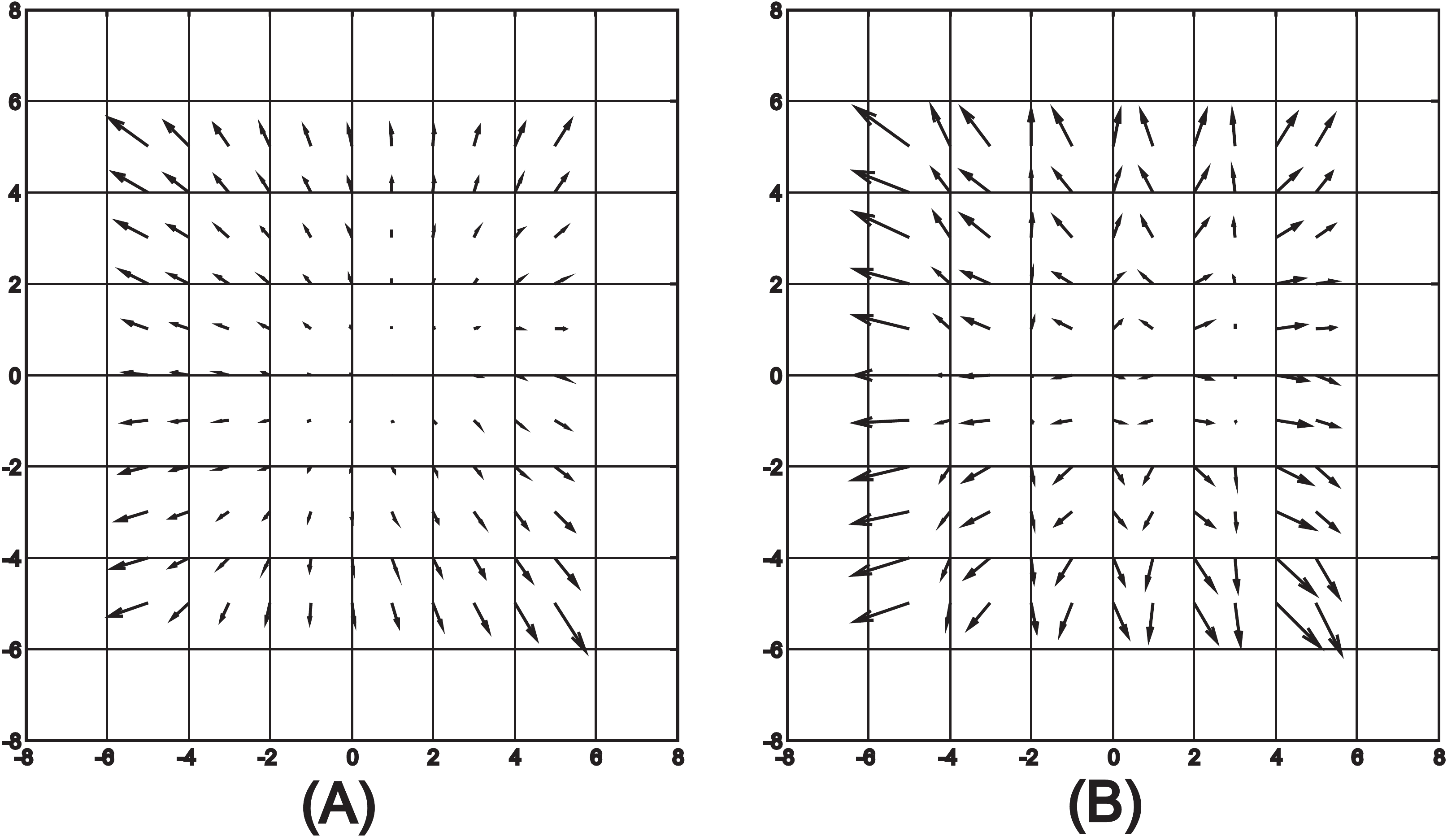} \caption{(A) Real Field
and (B) Recovered Field when the grid resolution is $N=2U/P=11$ and
the sources are closer to the bounded
domain}\label{fig:ManyPointSourcesCloseToDomain}
\end{figure}
\begin{figure}[!htb] \centering
\includegraphics[width=0.4\textwidth]{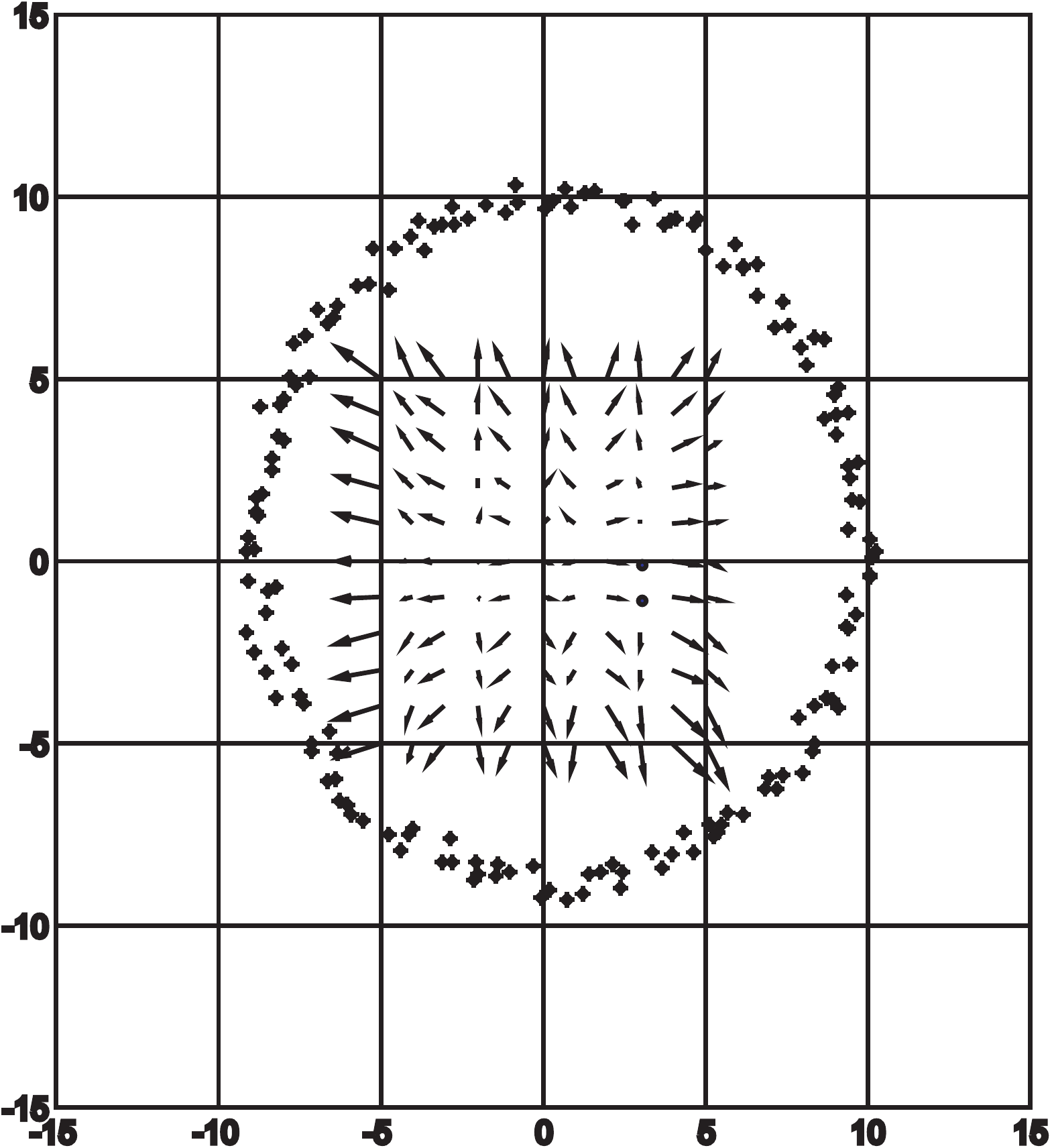} \caption{Dots depict the locations of the point
sources which are closer to the bounded domain and the arrows depict
the recovered vector field inside the bounded domain
}\label{fig:ManySourcesDistributionCloseToDomain}
\end{figure}
\newpage
\begin{figure}[!htb] \centering
\includegraphics[width=0.5\textwidth]{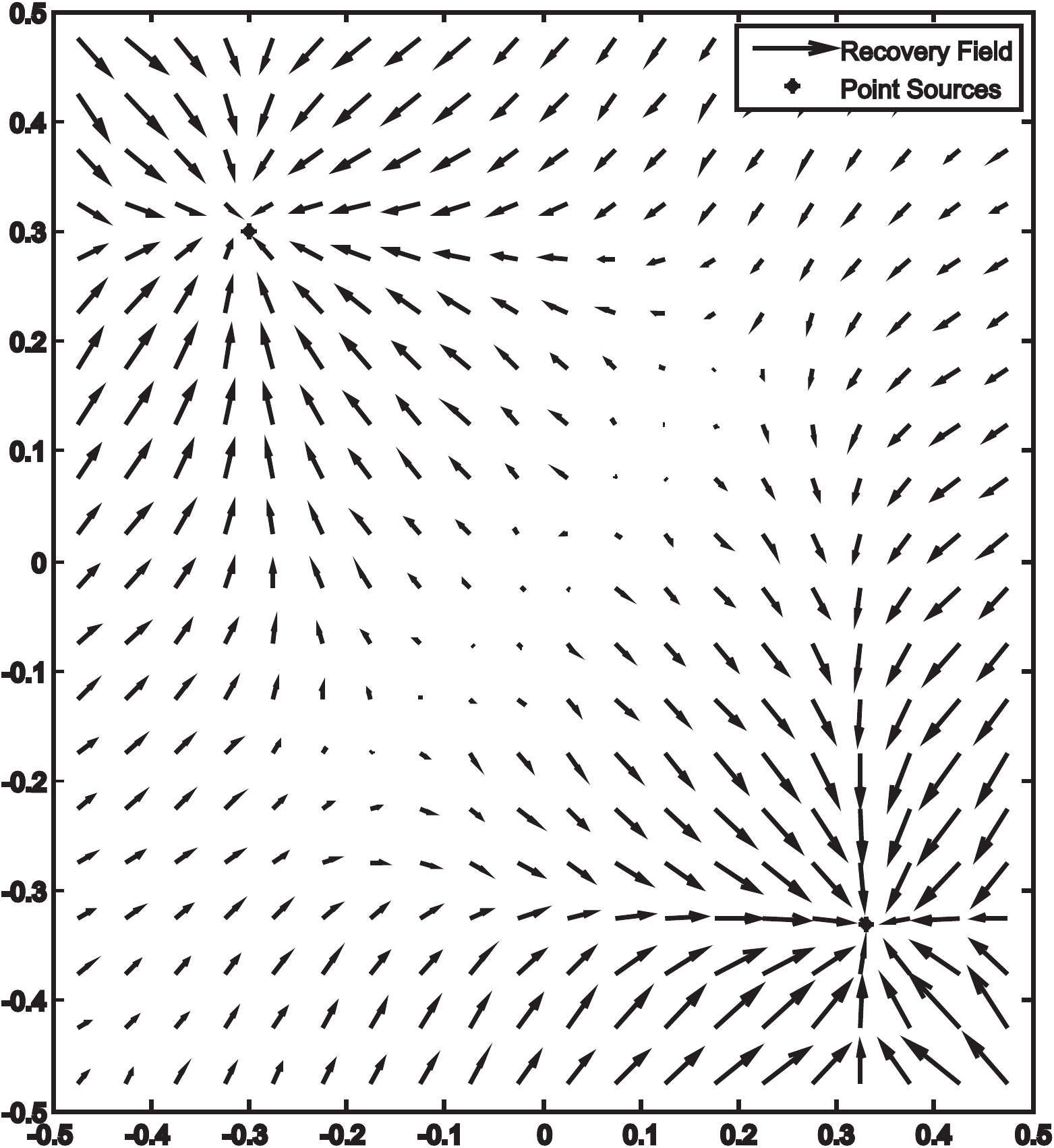} \caption{Recovered field when two point sources are inside the bounded domain}
\label{fig:SourcesInsideDomain}
\end{figure}
\item
 In the case where the sources are localized inside the bounded domain, the
 relative error of the magnitude of the field is
 too high while the phase error is low and the vector arrows point to the sources' positions.
 The field close to sources is not smooth, i.e. the real field has both high and low spatial frequency terms and thus
 for poor spatial resolution of the bounded domain (fig. \ref{fig:SourcesInsideDomain}): only the low
 spatial frequency terms can be recovered. The same actually happens in the previous case
 (fig.\ref{fig:ManySourcesDistributionCloseToDomain})
 when the point sources are placed closer to the grid. In
 figure \ref{fig:SourcesInsideDomain} the relative error of the magnitude is
 high due to low resolution. However, the recovered field points at the source locations which indicates that the low spatial
 frequency terms give some extremely useful information about the direction of the field.
\end{itemize}

\paragraph{Condition number vs Resolution\\}
Finally, figure \ref{fig:ConditionNumberVsResolution} depicts the
grid resolution against the condition number of the transfer matrix
of the linear system. By construction the transfer matrix of the
linear system depends on the geometrical characteristics of the
region where we intend to estimate the field and the change of the
resolution. So, as our experiments are performed in a rectangular
region then for constant resolution all the examples share the same
transfer matrix and what it changes in our linear system basically
is the boundary measurements $\textbf{b}$. Therefore, the following
graph presents the relationship between the rectangular resolution
and the conditioning of the transfer matrix.
\begin{figure}[!htb] \centering
\includegraphics[width=0.6\textwidth]{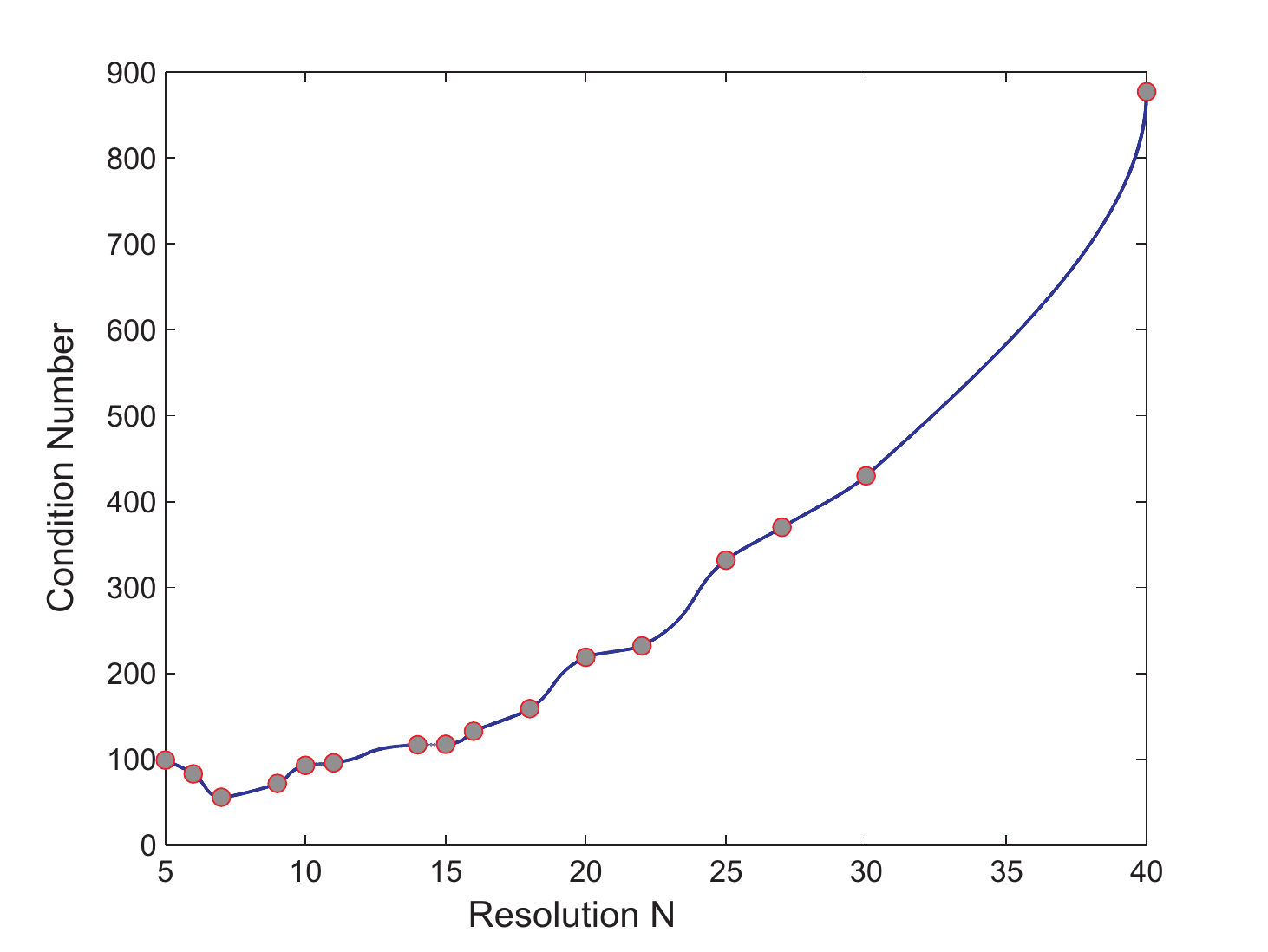}\caption{Condition number for
increasing values of grid resolution $N\times
N$.}\label{fig:ConditionNumberVsResolution}
\end{figure}
So, one more time we conclude that different discretization choices
of the grid impact the formulation of matrix $\textbf{A}$ and the
increase in spatial resolution actually increases the
ill-conditioning of this inverse problem.

\subsubsection{The effect of additive noise on the reconstruction}
In this paragraph we examine the effect of noise in $\textbf{b}$
measurements of the linear system
$\textbf{A}\textbf{x}_{LS}=\textbf{b}$ for the $2^{nd} example$ of
the previous subsection where the field was produced by $2$ point
sources with $Q=10^{-8}$ localized at $(-19,0)$ and $(19,19)$ far
from the bounded domain $\Omega:\{(x,y)\in[-2.5,2.5]^2\}$. According
to figure \ref{fig:ConditionNumberVsResolution} the lower acceptable
value of condition number $K(\textbf{A})$ is $72$ such as to balance
ill conditioning-resolution $N\times N=9\times 9$. So the majority
of the rest simulations with noisy measurements will be performed
for this resolution however, some further simulations with higher
condition number will be estimated in order to examine the stability
of the system in the presence of noise. Moreover, we employ the
relative magnitude of the input and output error to provide a
measure of stability factor $S$:
\begin{eqnarray*}
&e_{out}=\frac{\|\textbf{x}_{LS}-\textbf{x}_{exact}\|_2}{\|x_{exact}\|_2}\\
&e_{in}=\frac{\|\bar{\textbf{b}}-\textbf{b}\|_2}{\|\textbf{b}\|_2}\\
&S=\frac{e_{out}}{e_{in}}
\end{eqnarray*}

Noisy measurements are estimated according to
$\bar{\textbf{b}}=\textbf{b}+n\cdot randn()$ where $randn()$ is
Matlab function  which produce random numbers whose elements are
normally distributed with mean $0$, variance $\sigma^2=1$ and
$n\in[5 \cdot10^{-3},5\cdot 10^{-2}]$ In figure
\ref{fig:EinEoutExample2} for each $n$ value, we estimate the mean
value of $e_{out}$  using $100$ different noise vectors. Stability
factor was estimated between $8.5\sim9.5$ (we have to mention that
$S$ is meaningless for
$\bar{\textbf{b}\rightarrow\textbf{b}}$.)Particulary, x-axis depicts
the Signal to Noise ratios (SNR)of our inputs which are given by
$20log_{10}(1/mean(e_in)_n)$  while y-axis the mean values of
$mean(e_{out})_n$. Obviously for higher values of the SNR (low
additive noise) the output relative error is low.
\begin{figure}[!htb] \centering
\includegraphics[width=0.5\textwidth]{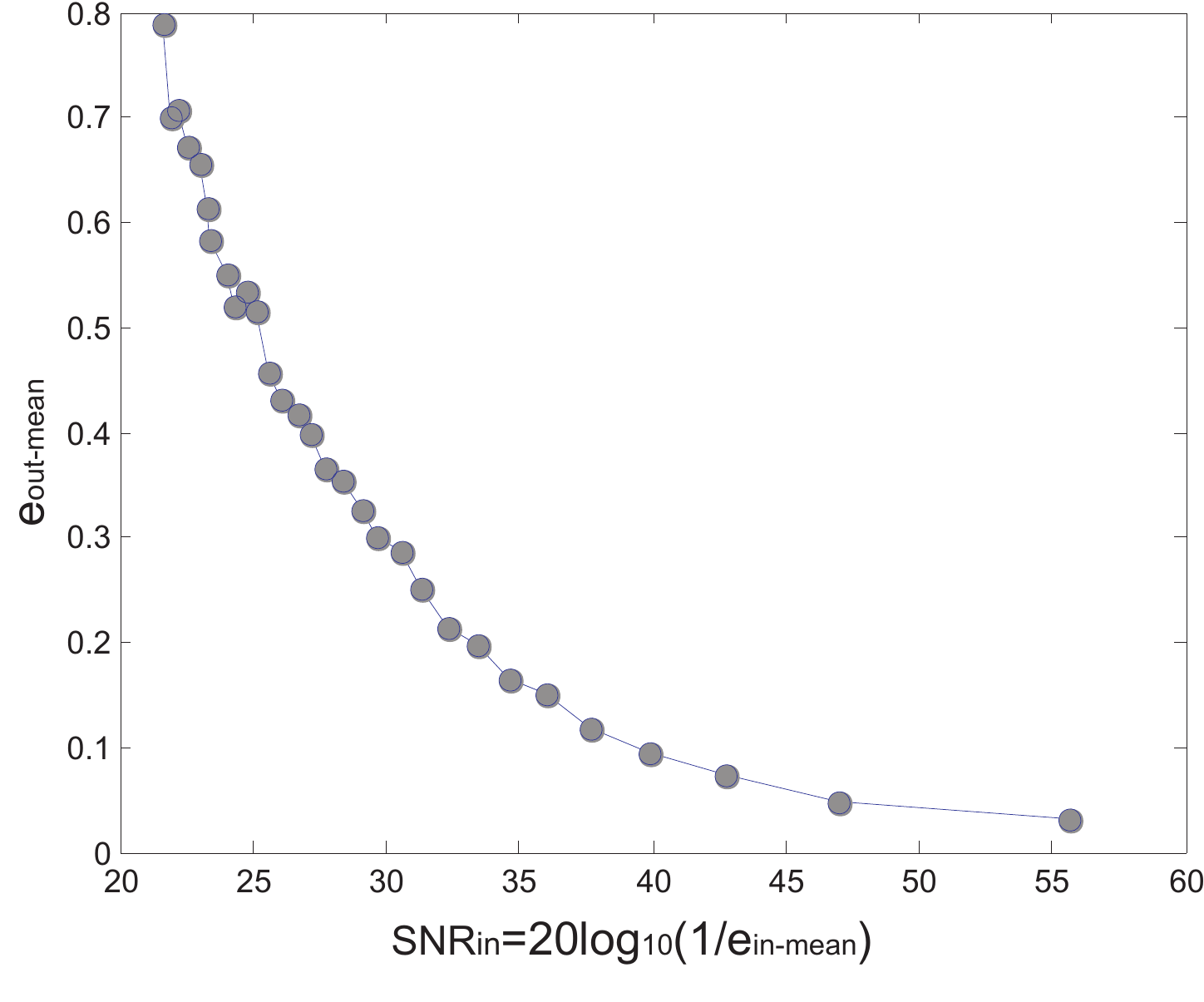} \caption{Mean values of the output error $e_{out}$ for
different values of the input SNR when the condition number of the
transfer matrix is low.}
 \label{fig:EinEoutExample2}
\end{figure}

For the case where the condition number is a bit higher i.e.
$\textbf{A}=132.6$ and $N\times N=16\times 16$ the curve between
$e_{out}$ and SNR is depicted in figure
\ref{figEinEoutExample2CondHigher}.
\begin{figure}[!htb] \centering
\includegraphics[width=0.5\textwidth]{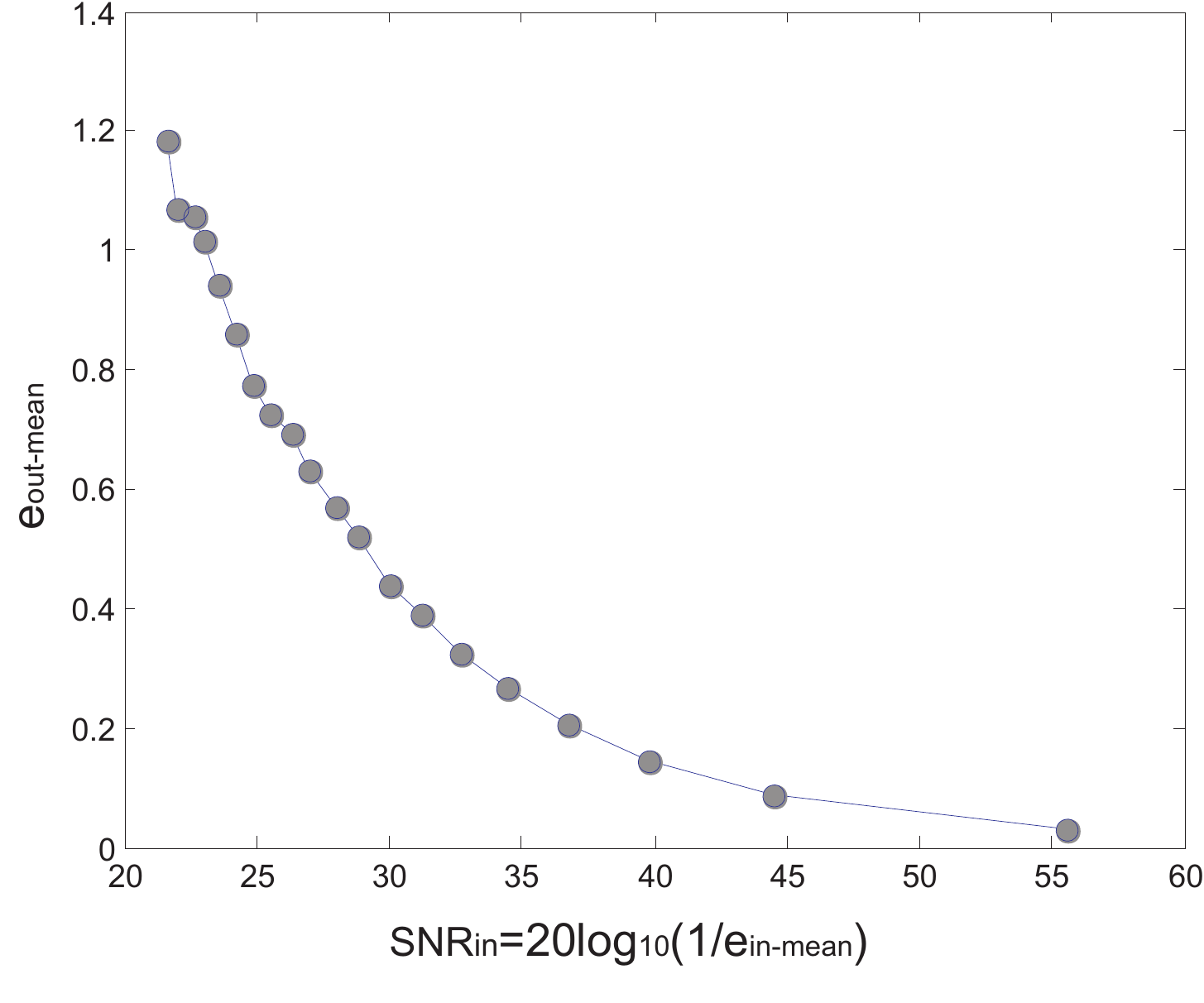} \caption{Mean values of the output error $e_{out}$ for
different values of the input SNR and $S$ factor range between
$14.1\sim18$.}
 \label{figEinEoutExample2CondHigher}
\end{figure}

Comparing figures \ref{fig:EinEoutExample2} and
\ref{figEinEoutExample2CondHigher} and stability factors $S$ we can
mention that the increase in the ill conditioning of the
problems(even a small increase) affects the accuracy of the output
solution in the presence of noise. Obviously when the level of noise
is high and we need a better accuracy(grid refinement) then extra
regularization such as Tikhonov should be considered and possible
pre-processing of the measurements to reduce noise level.

\subsection{Discussion}
Different discretization choices of the grid will impact the
formulation of matrix $\textbf{A}$. The increase in discretization
resolution on the bounded domain $\Omega$ actually increases the
ill-conditioned nature of the inverse problem. The ill conditioning
of the linear system results in the instability of the solution of
the system especially in the presence of noise. In the current
simulation, there was no additional noise. However, the majority of
real problems suffer from additive or multiplicative noise and a
future examination of the effect of noise to the inverse problem
will be performed.

In order to alleviate the ill conditioning of the linear system, one
may perform a coarse discretization of the bounded domain. This
leads to a ``rough'' approximation of the vector field (see fig.
\ref{fig:ManySourcesDistributionCloseToDomain}). For a non-smooth
field (fig.\ref{fig:SourcesInsideDomain}) the approximation error is
considerable. Particularly, let us consider vector field
$\textbf{E}(x,y)$ in bounded domain $\Omega$ as the sum of spatial
frequency terms according to Fourier series expansion
\[\textbf{E}(x,y)=\sum_{\textbf{k}}\textbf{e}_ke^{i\textbf{k}\langle x,y\rangle}+\bar{\textbf{e}}_ke^{-i\textbf{k}\langle x,y\rangle}\]
The spatial frequency $\textbf{k}$ is directly related to the
spatial resolution of the domain and basically the resolution
determines the spatial frequency band limit of the vector field to
be recovered. If the resolution of the domain is high then more
spatial frequency terms can be recovered. However, the
``uncontrolled'' increase in resolution worsens the problem
conditioning. So, the mathematical and experimental choice of the
optimal discretization of the bounded domain, such as to minimize
the approximation error while maximizing the stability of the linear
system needs to be assessed.

\section{Practical Aspects Of The Method}
A major objective of this section is to introduce the inverse
bioelectric problem as one of the most prominent biomedical
potential applications of the proposed vector field method.
Moreover, a comparison of the line integral method with the current
partial differential formulation will be presented.
\subsection{Inverse Bioelectric Field Problem}
The bioelectric field \cite{MalmivuoBiolelctroBook1995} is the
manifestation of the current densities inside the human body as a
result of the conversion of the energy from chemical to electric
form (excitation) in the living nerves, muscles cells and tissues in
general.

Bioelectric inverse techniques intend to estimate the source
distributions inside specific parts of the body, e.g brain and
heart, employing the EEG or ECG recordings (passive methods). The
ordinary mathematical modeling of this problem is based on the
solution of a boundary value problem of an elliptic partial
differential equation (PDE). As we will see later, this problem is
ill-posed as it does not satisfy some of the three Hadamard's
criteria: existence, uniqueness of the solution and continuous
dependence of the solution upon the data.
\subsection{PDE Methods}
The bioelectric field problem can be formulated in terms of either
the Poisson or the Laplace equation depending on the physical
properties of the problem \cite{BioengineeringHandbook}.

In particular, for the problem formulation, physical characteristics
of the electrical sources of the human body have to be defined
\cite{MalmivuoBiolelctroBook1995}. The primary source $\rho(x,y)$ of
electric activity is produced as a result of the transformation of
the energy inside the cells from chemical to electric, which
consequently induces an electric current of the form $\sigma
\textbf{E}$, where $\sigma$ is the bulk conductivity of the volume.
In general, the electric density is time-varying, however, the
passive way of the data acquisition (not imposed external potentials
like in electrical impedance tomography) using EEG or ECG to record
bioelectric source behavior in low frequencies (frequencies below
several kHz) \cite{BFP_johnson1993},\cite{Geselowitz} enables the
quasi-static treatment of the problem. The static consideration of
the field in a negligible short time implies that the capacitance
component of the tissues ($j\omega\textbf{E}$) can be assumed
negligible.

So, the total current density is given by
$\textbf{J}=\rho(x,y,z)+\sigma\textbf{E}$ and due to the
quasi-static condition it obeys $\nabla \cdot\textbf{J}=0$. A last
important point is that electromagnetic wave effects are also
neglected \cite{Geselowitz} and therefore the electric field is
given by $\textbf{E}=-\nabla\Phi$.

\subsubsection{Poisson Equation Formulation}
The commonly used inverse bioelectric methods try to calculate the
internal current sources $\textbf{J}$ given a subset of
electrostatic potentials measured on the scalp or other part of the
body, the geometry and electrical conductivity properties within
this human part.

The physical and mathematical model of the inverse bioelectric
problem is derived from $\nabla \cdot\textbf{J}=0$ and the scalar
potential representation of the field $\textbf{E}=-\nabla{\Phi}$ in
a bounded domain $\Omega$. Hence, the inverse problem is based on
the solution of a Poisson-like equation:
\begin{equation}\label{eq:PoissonPDE}\nabla\sigma\nabla\Phi=-\rho
\,\,\,\,\mbox{in}\,\,\Omega\end{equation} with the Cauchy boundary
conditions:
\[\Phi=\Phi_0\,\,\,\mbox{on}\,\,\Sigma\subseteqq\Gamma_T \,\,\,\,\,\,\,\mbox{and}\,\,\,\,\,\,\,\sigma\nabla\Phi\hat{\textbf{n}}=0 \,\,\,\mbox{on}\,\,\Gamma_T=\partial\Omega\,\]
where $\Phi$ is the electrostatic potential, $\sigma$ is the
conductivity tensor, $\rho$ are the current sources per unit volume,
and $\Gamma_T$ and $\Omega$ represent the surface and the volume of
the body part (e.g. head), respectively.

The Dirichlet condition
($\Phi=\Phi_0\,\,\,\mbox{on}\,\,\Sigma\subseteqq\Gamma_T$) is a
mathematical abstraction of potential measurements on $\Sigma$,
which are in reality obtained from only a finite number of
electrodes, and  the Neumann condition
($\sigma\nabla\Phi\hat{\textbf{n}}=0
\,\,\,\mbox{on}\,\,\Gamma_T=\partial\Omega$) is describing that no
current flows out of the body.

For the solution of the problem, the modeling of the sources is
extremely essential and the difficulties of the design of an
accurate model impose significant limitations to the current
solution techniques.

In particular, usually the sources are mathematically assumed to be
dipoles with unknown magnitude, position, and orientation
\cite{BFP_johnson1993},\cite{BioengineeringHandbook}. The simplest
problem uses the assumption that $\rho = Q(m, P, O)$ is a single
dipole, characterized by  three parameters (six degrees of freedom),
corresponding to magnitude $m$, position $P (x, y, z)$, and
orientation $O$. These parameters are adjusted such that the
resulting electrostatic potential best matches the given measured
data (trial and error method)\cite{BFP_johnson1993}. More general
models with several dipoles and consequently increased number of
unknown parameters have been designed. Unfortunately, more general
cases where there are no previous assumptions about the sources have
no unique solution. Solutions can only be found with suitable
regularization or a restriction of the solution space by assuming
$\rho$ of a special form.

Generally, the inverse bioelectric field problem
(\ref{eq:PoissonPDE}) is ill posed as for an arbitrary source
distribution, it does not have a unique solution and the solution
does not depend continuously on the data i.e. small errors in
measurements may cause large errors in the solution.

More specifically, according to Helmholtz's theorem earlier, and
Thevenin's (Norton's) theorem later
\cite{Johnson03},\cite{MalmivuoBiolelctroBook1995}, it is always
possible to replace a combination of sources and associated circuity
with a single equivalent source and a series of impedances. Thus,
circuits with different structure (impedances' sequence and source
location) can give the same equivalent circuit (Thevenin circuit).
Similarly, in the electric field which can be assumed as an
equivalent representation of a circuit, a different source
distribution can give the same equivalent source-generator
(Superposition Principle) which gives rise to the observed weighted
boundary potentials. In other words, the same observed potential
measurements could be produced from different source distributions
and thus the problem has not unique solution. So, for the estimation
of the sources the common method \cite{MalmivuoBiolelctroBook1995}
is to breaks up the solution domain into a finite number of
subdomains. In each subdomain, a simplified model of the actual
bioelectric source (such as a dipole) is assumed. The problem then
is to find the magnitude  and direction of each of the simplified
sources in the subdomains.

In addition, theoretically the elliptic equations boundary value
problem of Poisson PDE with Cauchy boundary conditions leads to an
over-specified or too ``sensitive'' boundary value problem
\cite{MethodsForPhysics} resulting large errors in the solution for
small changes in boundary conditions. So, regularization techniques
should be performed for the stability of the solution \cite{Wang}.

\subsubsection{Laplace Equation Formulation}
If the inverse problem is formulated such as there are no sources on
the bounded domain, then we obtain the homogeneous counterpart of
the Poisson equation. These are the cases where there is interest on
the estimation of the potential distribution on a surface (e.g.
cortex or surface of the heart) and not in the whole volume from
boundary measurements (like scalp (EEG) or torso (ECG) recordings)
and hence the sources are out of the region of interest.

So, the Cauchy boundary value problem is formed by Laplacian
equation $\nabla \sigma\nabla\Phi=0$. This is also an ill posed
problem in the sense that the solution does not depend continuously
on the data. Numerically, the solution of the Laplace's equation can
be formulated as $\textbf{A}\Phi=\Phi_{boundaries}$. As the
continuous problem is ill posed due to its physical nature
\cite{Wang},\cite{MethodsForPhysics}, matrix $\textbf{A}$ will be
ill conditioned.

There are different regularization methods and numerical techniques
for the approximate solution of this system
\cite{Johnson97},\cite{Wang}. In the current section, we will not
give further details about these techniques. However, a brief
comparison between Laplace's equation and the line integrals methods
for the estimation of a vector field will be presented.

\subsection{Vector Field Method}
In the previous description, the inverse bioelectric theory intends
to estimate the positions and strengths of the sources within a
volume (e.g brain) that could have given rise to the observed
potential recordings under the assumptions that the field is
quasi-static and irrotational. Obviously, the physical principles of
the bioelectric field are fitted well with the preconditions and
assumptions of the proposed vector field tomographic method. Thus,
the proposed method can be applied to recover a bioelectric field in
discrete points using the boundary information without needing any
prior information about the sources' formulation while the boundary
conditions are automatically satisfied as they are directly
incorporated in the line integral formulation.

\paragraph{\emph{Vector Field Recovery with Laplace and Vector Field Method}\\}
We present the recovery of vector field in a bounded domain
$[-5.5,5.5]^2$ with cell size $P\times P=1\times1$, grid $N\times
N=11\times11$ and sampling step $\Delta r=0.0001P$, solving two
systems of linear equations where one was derived from the line
integral method as it was described previously and another from the
numerical solution of the Laplace equation
\[\nabla^2\Phi=\frac{\partial^2\phi}{\partial x^2}+\frac{\partial^2\phi}{\partial y^2}=0\]
based on the finite difference approximation
\begin{equation*}\label{SecondOrderXX}\frac{\partial^2 \phi}{\partial
x^2}=\frac{\phi_{i+1,j}-2\phi_{i,j}+\phi_{i-1,j}}{h^2} \,\,\,
\frac{\partial^2 \phi}{\partial
y^2}=\frac{\phi_{i,j+1}-2\phi_{i,j}+\phi_{i,j-1}}{h^2}
\end{equation*}
\begin{figure}[!htb]
\centering
\includegraphics[width=0.8\textwidth]{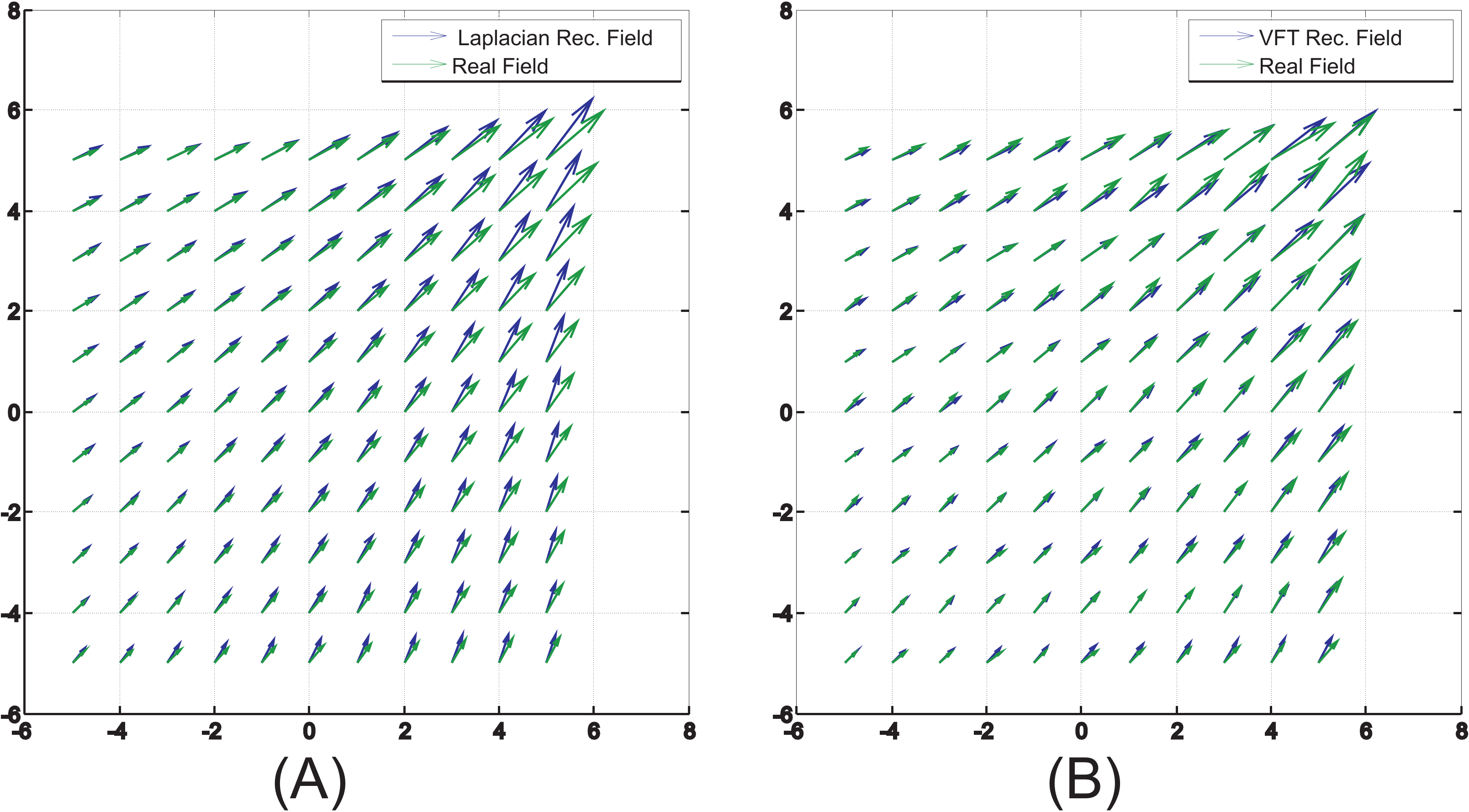} \caption{ The recovered fields applying (A) numerical solution of a Laplacian equation and (B)
vector field line integrals, produced by a point source $Q=10^{-8}$
localized in $(13,15)$.}\label{fig:VFTvsLaplace}
\end{figure}

For the recovery of the field in figure \ref{fig:VFTvsLaplace}A, the
rows of system $\textbf{A}\Phi=\Phi_{boundaries}$ were defined
according to
\[4\phi_{i,j}-\phi_{i-1,j}-\phi_{i+i,j}-\phi_{i,j-1}-\phi_{i,j+1}=0\]
with known boundary potential values
$\phi_{0,j},\phi_{i,0},\phi_{0,N}$ and $\phi_{N,j}$. After the
estimation of $\phi$ the vector field was calculated by
\[\textbf{E}=\left(
\frac{\phi_{i+1,j}-\phi_{i-1,j}}{h},\frac{\phi_{i,j+1}-\phi_{i,j-1}}{h}\right)\]
with $h=P$. The coarse recovery of an electrostatic field employing
the vector field method presented better and more accurate results
than the finite difference Laplacian solution. Here, we have to
consider that $\textbf{E}$ is estimated from the potential values
$\Phi$ and thus extra numerical errors were introduce to the
Laplacian solution.

\chapter{Conclusions}
Over the past two decades there has been an increasing interest in
vector field tomographic methods and many mathematical approaches
based on ray and Radon Transform theory have been reported.
Nowadays, there is a lack of potential applications, although all
the theoretical findings and the previous work may be extremely
essential for the evolution and the implementation of robust
techniques. Particulary, the theoretical and experimental study of
the proposed vector field method for real problems with simple
physical and geometric properties is the first step towards the
deeper understanding and the evolution of the method for more
complex real applications. Briefly, our future work will be focused
on three basic tasks:
\begin{itemize}
  \item robust theoretical formulation of the vector field method
  according to inverse problems theory;
  \item more sophisticated numerical approaches based on finite elements theory and further experiments and
  simulations;
  \item practical aspects of the method especially in the field of
  inverse bioelectric field problem with comparisons between the
  vector field method and the current PDE methods.
\end{itemize}

\section{Future Work}
\begin{itemize}
  \item\textbf{Theoretical Foundation}\\
  In subsection \ref{subsection:IllPosednessConsideration},
  basic physics principles are used to describe the ill conditioning
  of the vector field method and a more intuitive consideration of the
  ill posedness of the problem was given.
  Next step is the robust mathematical formulation of this inverse problem in Hadamard's sense, examining
  the uniqueness of the solution and the stability of the solution to the noise in the input
  data.

  In the mathematical theory of the inverse problems, an inverse problem is
  described by the equation \[K\textbf{x}=y\] where given $y$ and $K$ we have to solve
 for $\textbf{x}$.

  The vector field method is an inverse problem where operator $K$ is an integral operator.
  This integral operator is compact in many natural topologies and in general
  it is not one-to-one \cite{Kirsch}. Almost always compact operators $K$ are ill posed.
  However, this basic knowledge is not enough for the mathematical
  interpretation of the ill posedness of the vector field problem,
  so further theorems of the inverse problem theory,
  and assumptions should be considered.
  A robust and complete mathematical formulation of the problem is essential for
  a better understanding of the problem's nature and the potential numerical
  adaptations of this theoretical problem to real applications.
  \item \textbf{Numerical Implementation and Further Experiments}\\
  The numerical solution of the vector field recovery method is based on the
  discretization of the problem. Ill posedness of the continuous
  problem is expressed as ill conditioning of the system of linear
  equations in the discrete version. The conditioning of the problem
  depends on the resolution of the discrete domain and the applied interpolations schemes.
  A ``naive'' discretization may lead to disastrous results as
  for  very ``coarse'' discretization, the approximation errors increase
  while when an extreme grid refinement is applied, the conditioning of the
  numerical system deteriorates.

  In particular,
   the grid refinement can be used to increase the resolution in
   bounded domains, but must be used with care because refined meshes
   worsen the conditioning of the inverse problem and simultaneously
   increase the computational cost.
   So, there is an open
   question about the optimal discretization of the bounded domain,
   such as to minimize the approximation error while maximizing the
   stability of the linear system. For this purpose, discretization error bounds should be
   estimated mathematically and experimentally.
   Moreover, the discretization of the domain employing other
   elements,
   not cells or tiles as in the square domain, embedding more sophisticated techniques should be considered.

    Another option to overcome the ill conditioning of the problem is by increasing the
    observed data available on the boundary (increase boundary resolution).
    For the case of a simple $2D$ square domain,
    there is an initial approach for the definition of the optimum number and positions
    of the boundary measurements reported in \cite{SamplingBounds}. However, further investigation and
    generalization for arbitrary shapes of the bounded domain need to be
    considered. Moreover, it should be mentioned that the
    increase of boundary measurements is not a very efficient way to
    solve the problem as in practical applications (e.g. EEG) there are
    only a few measurements and from these few
    measurements we have to extract as much information as possible. Finally, the increased boundary resolution, not as a
    consequence of the increasing measured data but as a results of
    interpolation schemes \cite{Wang}, does not improve the
    problem conditioning in real problems.

    \begin{figure}[!htb] \centering
    \includegraphics[width=0.4\textwidth]{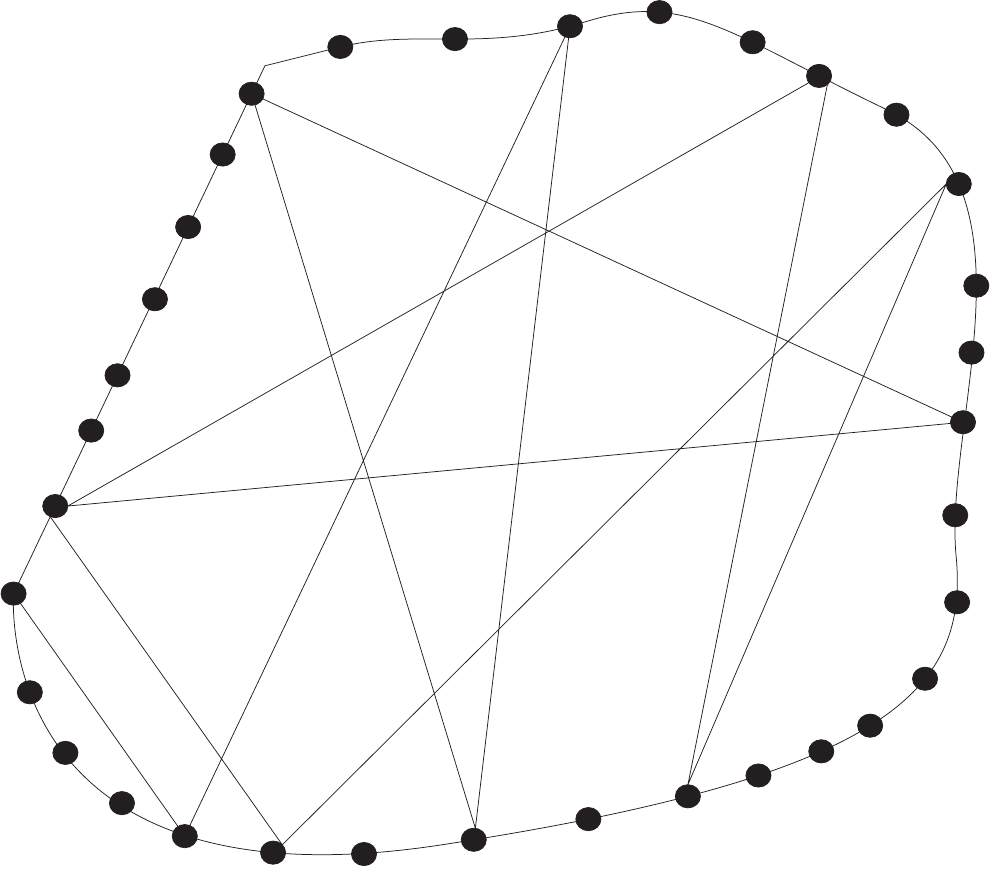} \caption{Recovery of the field
    using the vector field method for an arbitrary boundary shape.} \label{fig:SourcesInsideDomain}
    \end{figure}

\item
After the theoretical and the numerical schemes have been completed,
the
next steps are the \textbf{EEG application}.\\
In ordinary EEG source analysis, the inverse problem estimates the
positions and strengths of the sources within the brain that could
have given rise to the observed scalp potential recordings. In the
current project, the final stage is the implementation of a
different approach for the EEG analysis employing the proposed
method. Rather than performing strengths estimation or location of
the sources inside the brain, which are very complicated tasks, a
reconstruction of the corresponding static bioelectric field will be
performed based on the line integral measurements. This bioelectric
field can be treated as an ``effective'' equivalent state of the
brain. Thus for instance, health conditions or a specific pathology
(e.g. seizure disorder) may be recognized. This can be achieved by
designing a slice-by-slice reconstruction algorithm for vector field
recovery or the $3D$ extension  of the proposed theoretical model
where new numerical problems concerning the design of huge linear
systems (sparse matrices and different computational strategies)
arise.
\end{itemize}


\appendix
\chapter{Rank of Perturbed Matrix $\bar{\textbf{A}}$}\label{Appen:RankOfPerturbedMatrix}
As it is described in \cite{MatrixAndLinearAlgebra} $p.216$, the
rank of a perturbed matrix $\bar{\textbf{A}}=\textbf{A}-\textbf{E}$
is equal to or greater than the $rank(\textbf{A})$ when the entries
of $\textbf{E}$ have small magnitude. According to the proof on
$p.217$ \cite{MatrixAndLinearAlgebra} matrix $\textbf{A}$ with
$Rank(\textbf{A})=p$ can be transformed in an equivalent matrix
(rank normal form) when elementary row and column operators are
applied i.e.
\[\textbf{A}\sim
\left(\begin{array}{cc} \textbf{I}_p & \textbf{0}\\
 \textbf{0}& \textbf{0}
 \end{array}\right)\]

 Equivalently for the $\textbf{A}-\textbf{E}$ where
 $\textbf{E}=\left(\begin{array}{cc}
 \textbf{E}_{11}^{p\times p}& \textbf{E}_{12}\\
\textbf{E}_{21} & \textbf{E}_{22}
 \end{array}\right)$ we obtain

 \[\textbf{A}-\textbf{E}\sim
\left(\begin{array}{cc} \textbf{I}_p-\textbf{E}_{11}^{p\times p} & \textbf{0}\\
 \textbf{0}& \textbf{S}
 \end{array}\right)\]

 Thus,
\begin{equation*}
\begin{split}
 Rank(\textbf{A}-\textbf{E})&=Rank(\textbf{I}_p-\textbf{E}_{11}^{p\times
 p})+Rank(\textbf{S})=\\
&=Rank(\textbf{A})+Rank(\textbf{S})\geqslant Rank(\textbf{A})
\end{split}
\end{equation*}

In practical problem, $\textbf{S}=\textbf{0}$ is very unlikely as
the perturbation (discretization or roundoff errors) is almost
``random''. Therefore
$Rank(\textbf{A}-\textbf{E})>Rank(\textbf{A})$.

\bibliographystyle{plain}
\bibliography{IC_transferPHD_Koulouri_2011}

\end{document}